\documentclass[11pt,a4paper,leqno]{article}
\usepackage[english]{babel}
\usepackage[latin1]{inputenc}

\usepackage{amsfonts}
\usepackage{amsmath}
\usepackage{amsthm}
\usepackage{bm}
\usepackage[round]{natbib}
\usepackage{graphicx}
\usepackage{caption}
\usepackage{appendix}
\usepackage{subcaption}
\usepackage[parfill]{parskip}
\usepackage{epstopdf}
\usepackage{verbatim}
\usepackage{booktabs}
\usepackage{multirow}
\usepackage[colorlinks=true,citecolor=blue]{hyperref}

\makeatletter
\def\thm@space@setup{%
  \thm@preskip=\parskip \thm@postskip=0pt
}
\makeatother

\newcommand{\dd}{{\mathrm d}}

\newcommand{\e}{{\rm e}}
\newcommand{\E}{{\mathbb E}}
\newcommand{\Pa}{{\mathbb P}}
\newcommand{\Q}{{\mathbb Q}}

\newcommand{\R}{{\mathbb R}}
\newcommand{\N}{{\mathbb N}}

\newcommand{\Fcal}{{\mathcal F}}
\newcommand{\Gcal}{{\mathcal G}}

\newcommand{\Scal}{{\mathcal S}}

\newtheorem{proposition}{Proposition}[section]
\newtheorem{lemma}[proposition]{Lemma}

\newtheorem{remark}[proposition]{Remark}

\newtheorem{example}[proposition]{Example}

\title{A Term Structure Model for Dividends and Interest Rates\footnote{We thank participants at the Workshop on Dynamical Models in Finance in Lausanne, the 8th General Advanced Mathematical Methods in Finance conference in Amsterdam, the 2nd International Conference on Computational Finance in Lisbon, the 11th Actuarial and Financial Mathematics Conference in Brussels, the 2018 Swiss Finance Institute Research Days in Gerzensee, the 10th Bachelier World Congress in Dublin, the 2018 Young Researchers Workshop on Data-Driven Decision Making at Cornell University, the 2019 Cambridge-Lausanne workshop,  and seminars at McMaster University, New York University, Princeton University, UC Berkely, and University College Dublin, as well as Peter Carr, J\'er\^ome Detemple (discussant), Alexey Ivashchenko (discussant), Martin Lettau, Chris Rogers (discussant), Radu Tunaru, and three anonymous referees for their comments. The research leading to these results has received funding from the European Research Council under the European Union's Seventh Framework Programme (FP/2007-2013) / ERC Grant Agreement n.~307465-POLYTE.}}

\author{Damir Filipovi\'c\footnote{EPFL and Swiss Finance Institute. Email: damir.filipovic@epfl.ch} \quad \quad Sander Willems\footnote{EPFL and Swiss Finance Institute. Email: willems.sander@gmail.com}}

\date{May 22, 2020\\
forthcoming in \textit{Mathematical Finance}}
\providecommand{\noopsort}[1]{}

\begin{document}
\maketitle

\begin{abstract}
Over the last decade, dividends have become a standalone asset class instead of a mere side product of an equity investment. We introduce a framework based on polynomial jump-diffusions to jointly price the term structures of dividends and interest rates. Prices for dividend futures, bonds, and the dividend paying stock are given in closed form. We present an efficient moment based approximation method for option pricing. In a calibration exercise we show that a parsimonious model specification has a good fit with Euribor interest rate swaps and swaptions, Euro Stoxx 50 index dividend futures and dividend options, and Euro Stoxx 50 index options.

\medskip



\end{abstract}

\section{Introduction}
In recent years there has been an increasing interest in trading derivative contracts with a direct exposure to dividends. \cite{brennan1998stripping} argues that a market for dividend derivatives could promote rational pricing in stock markets. In the over-the-counter (OTC) market, dividends have been traded since 2001 in the form of dividend swaps, where the floating leg pays the dividends realized over a predetermined period of time. The OTC market also accommodates a wide variety of more exotic dividend related products such as knock-out dividend swaps, dividend yield swaps and swaptions. Dividend trading gained significant traction in late 2008, when Eurex launched exchange traded futures contracts referencing the dividends paid out by constituents of the Euro Stoxx 50. The creation of a futures market for other major indices (e.g., the FTSE 100 and Nikkei 225) followed shortly after, as well as the introduction of exchange listed options on realized dividends with maturities of up to ten years.
Besides the wide variety of relatively new dividend instruments, there is another important dividend derivative that has been around since the inception of finance: a simple dividend paying stock. Indeed, a share of stock includes a claim to all the dividends paid over the stock's lifetime. Any pricing model for dividend derivatives should therefore also be capable of efficiently pricing derivatives on the stock paying the dividends. What's more, the existence of interest rate-dividend hybrid products, the relatively long maturities of dividend options, and the long duration nature of the stock all motivate the use of stochastic interest rates. Despite its apparent desirability, a tractable joint model for the term structures of interest rates and dividends, and the corresponding stock, has been missing in the literature to date.

We fill this gap and develop an integrated framework to efficiently price derivatives on dividends, stocks, and interest rates. We first specify dynamics for the dividends and discount factor, and in a second step we recover the stock price in closed form as the sum of the fundamental stock price (present value of all future dividends) and possibly a residual bubble component.
The instantaneous dividend rate is a linear function of a multivariate factor process. The interest rates are modeled by directly specifying the discount factor to be linear in the factors, similarly as in \cite{filipovic2014linear}. The factor process itself is specified as a general polynomial jump-diffusion, as studied in \cite{Filipovic2017}. Such a specification makes the model tractable because all the conditional moments of the factors are known in closed form. In particular, we have closed form expressions for the stock price and the term structures of dividend futures and interest rate swaps. Any derivative whose discounted payoff can be written as a function of a polynomial in the factors is priced through a moment matching method. Specifically, we find the unique probability density function with maximal  Boltzmann-Shannon entropy matching a finite number of moments of the polynomial, as in \cite{mead1984maximum}. We then obtain the price of the derivative by numerical integration. In particular, this allows us to price swaptions, dividend options, and options on the dividend paying stock. We show that our polynomial framework also allows to incorporate seasonal behavior in the dividend dynamics.

Within our polynomial framework, we introduce the \emph{linear jump-diffusion} (LJD) model. We show that the LJD model allows for a flexible dependence structure between the factors. This is useful to model a dependence between dividends and interest rates, but also to model the dependence within the term structure of interest rates or dividends. We calibrate a parsimonious specification of the LJD model to market data on Euribor interest rate swaps and swaptions, Euro Stoxx 50 index dividend futures and dividend options, and Euro Stoxx 50 index options. Our model reconciles the relatively large implied volatility of the index options with the relatively small implied volatility of dividend options and swaptions through a negative correlation between dividends and interest rates. The successful calibration of the model to three different classes of derivatives (interest rates, dividends, and equity) illustrates the high degree of flexibility offered by our framework.

Our paper is related to various strands of literature. In the literature on stock option pricing, dividends are often assumed to be either deterministic (e.g., \citet{Bos}, \citet{Bos2}, \citet{vellekoop2006efficient}), a constant fraction of the stock price (e.g., \citet{merton1973theory}, \cite{korn2005stocks}), or a combination of the two (e.g., \cite{kim1995alternative}, \cite{overhaus2007equity}). \cite{geske1978pricing}  and \cite{lioui2006black} model dividends as a stochastic fraction of the stock price. They derive Black-Scholes type of equations for European option prices, however dividends are not guaranteed to be non-negative in both setups.  \cite{chance2002european} directly specify log-normal dynamics for the $T$-forward price of the stock, with $T$ the maturity of the option. Closed form option prices are obtained as in \cite{black1976pricing}, assuming that today's $T$-forward price is observable. This approach is easy to use since it does not require any modeling assumptions on the distribution of the dividends. However, it does not produce consistent option prices for different maturities. \cite{bernhart2015consistent} take a similar approach, but suggest to fix a time horizon $T$ long enough to encompass all option maturities to be priced. The $T$-forward price is modeled with a non-negative martingale and the stock price is defined as the $T$-forward price plus the present value of dividends from now until time $T$. As a consequence, prices of options with maturity smaller than $T$ will depend on the joint distribution between future dividend payments and the $T$-forward price, which is not known in general. \cite{bernhart2015consistent} resort to numerical tree approximation methods in order to price options. The dependence of their model on a fixed time horizon still leads to time inconsistency, since the horizon will necessarily have to be extended at some point in time. We contribute to this literature by building a stock option pricing model that guarantees non-negative dividends, is time consistent, and remains tractable.

Another strand of literature studies stochastic models to jointly price stock and dividend derivatives. \citet{buehler2010stochastic} assumes that the stock price jumps at known dividend payment dates and follows log-normal dynamics in between the payment dates. The jump amplitudes are driven by an Ornstein-Uhlenbeck process such that the stock price remains log-normally distributed and the model has closed form prices for European call options on the stock. The high volatility in the stock price is reconciled with the low volatility in dividend payments by setting the correlation between the Ornstein-Uhlenbeck process and the stock price extremely negative ($-95\%$). A major downside of the model is that dividends can be negative. Moreover, although  the model has a tractable stock price, the dividends themselves are not tractable and Monte-Carlo simulations are required to price the dividend derivatives. In more recent work, \cite{buehler2015volatility} decomposes the stock price in a fundamental component and a residual bubble component. The dividends are defined as a function of a secondary driving process that mean reverts around the residual bubble component. This model has closed form expressions for dividend futures, but Monte-Carlo simulations are still necessary to price nonlinear derivatives. \cite{guennoun2017equity} consider a stochastic local volatility model for the pricing of stock and dividend derivatives. Their model guarantees a perfect fit to observed option prices, however all pricing is based on Monte-Carlo simulations.
\cite{tunaru2017dividend} proposes two different models to value dividend derivatives. The first model is similar to the one of \cite{buehler2010stochastic}, but models the jump amplitudes with a beta distribution. This guarantees positive dividend payments. However, the diffusive noise of the stock is assumed independent of the jump amplitudes in order to have tractable expressions for dividend futures prices. Smoothing the dividends through a negative correlation between stock price and jump amplitudes, as in \cite{buehler2010stochastic}, is therefore not possible. In a second approach, \cite{tunaru2017dividend} directly models the cumulative dividends with a diffusive logistic growth process. This process has, however, no guarantee to be monotonically increasing, meaning that negative dividends can occur frequently. \cite{willems2019linear} jointly specifies dynamics for the stock price and the dividend rate such that the stock price is positive and the dividend rate is a non-negative process mean-reverting around a constant fraction of the stock price. The model of \cite{willems2019linear} is in fact a special case of the general framework introduced in our paper, although it is different from the LJD model and does not incorporate stochastic interest rates.
We add to this literature by allowing for stochastic interest rates, which is important for the valuation of interest rate-dividend hybrid products or long-dated dividend derivatives (e.g., the dividend paying stock). Our model produces closed form prices for dividend futures and features efficient approximations for option prices which are significantly faster than Monte-Carlo simulations. In particular, we give an example of a hybrid option on the dividend-interest rate spread that can be priced efficiently in our framework. The low volatility in dividends and interest rates is reconciled with the high volatility in the stock price through a negative correlation between dividends and interest rates. 

Our work also relates to literature on constructing an integrated framework for dividends and interest rates. Previous approaches were mainly based on affine processes, see e.g. \cite{bekaert1999stock}, \cite{mamaysky2002joint}, \cite{d2006international}, \citet{lettau2007long, lettau2011term}, and \cite{lemke2009term}. In more recent work, \cite{kragt2014dividend} extract investor information from dividend derivatives by estimating a two-state affine state space model on stock index dividend futures in four different stock markets. Instead of modeling dividends and interest rates separately, they choose to model dividend growth, a risk-free discount rate, and a risk premium in a single variable called the `discounted risk-adjusted dividend growth rate'. \cite{yan2014estimating} uses zero-coupon bond prices and present value claims to dividend extracted from the put-call parity relation to estimate an affine term structure model for interest rates and dividends. \cite{suzuki2014measuring} uses a Nelson-Siegel approach to estimate the fundamental value of the Euro Stoxx 50 using dividend futures and Euribor swap rates. We add to this literature by building an integrated framework for dividends and interest rates using the class of polynomial processes, which contains the traditional affine processes as a special case.

Finally, our work also relates to literature on moment based option pricing.
\cite{jarrow1982approximate}, \cite{corrado1996s}, and \cite{collin2002pricing} use Edgeworth expansions to approximate the density function of the option payoff from the available moments. Closely related are Gram-Charlier expansions, which are used for option pricing for example by \cite{corrado1996skewness}, \cite{jondeau2001gram}, and \cite{ackerer2016jacobi}. Although these series expansions allow to obtain a function that integrates to one and matches an arbitrary number of moments by construction, it has no guarantee to be positive. In this paper, we find the unique probability density function with maximal Boltzmann-Shannon entropy. subject to a finite number of moment constraints. Option prices are then obtained by numerical integration. A similar approach is taken by \cite{fusai2002accurate} to price Asian options. The principle of maximal entropy has also been used to extract the risk-neutral distribution from option prices, see e.g.\ \cite{buchen1996maximum}, \cite{jackwerth1996recovering}, \cite{avellaneda1998minimum}, and \cite{rompolis2010retrieving}. There exist many alternatives to maximizing the entropy in order to find a density function satisfying a finite number of moment constraints. For example, one can maximize the smoothness of the density function (see e.g., \cite{jackwerth1996recovering}) or directly maximize (minimize) the option price itself to obtain an upper (lower) bound on the price (see e.g., \cite{lasserre2006pricing}). A comparison of different approaches is beyond the scope of this paper.

The remainder of the paper is structured as follows. Section \ref{section:pol_FW} introduces the factor process and discusses the pricing of dividend futures, bonds, and the dividend paying stock. In Section \ref{section:option_pricing} we explain how to efficiently approximate option prices using maximum entropy moment matching. Section \ref{section:LJD} describes the LJD model. In Section \ref{section:numerical} we calibrate a parsimonious model specification to real market data. Section \ref{section:extensions} discusses some extensions of the framework. Section \ref{section:conclusion} concludes. All proofs and technical details can be found in the appendix.

\section{Polynomial framework}\label{section:pol_FW}
We consider a financial market modeled on a filtered probability space $(\Omega,\Fcal,\Fcal_t,\Q)$ where $\Q$ is a risk-neutral pricing measure. Henceforth $\E_t[\cdot]$ denotes the $\Fcal_t$-conditional expectation. We model the uncertainty in the economy through a factor process $X_t$ taking values in some state space $E\subseteq \R^d$.\footnote{We assume that $E$ has non-empty interior.} We assume that $X_t$ is a polynomial jump-diffusion (cfr.\ \cite{Filipovic2017}) with dynamics
\begin{equation}
\dd X_t = \kappa(\theta - X_t)\,\dd t +\dd M_{t},
\label{eq:factor_process}
\end{equation}
for some parameters $\kappa\in \R^{d\times d}$, $\theta\in\R^d$, and some $d$-dimensional martingale $M_t$ such that the generator $\Gcal$ of $X_t$ maps polynomials to polynomials of the same degree or less. One of the main features of polynomial jump-diffusions is the fact that they admit closed form conditional moments. For $n\in\N$, denote by $\mathrm{Pol}_n(E)$ the space of of polynomials on $E$ of degree $n$ or less and denote its dimension by $N_n$.\footnote{Since the interior of $E$ is assumed to be non-empty, $\mathrm{Pol}_n(E)$ can be identified with $\mathrm{Pol}_n(\R^d)$ and therefore $N_n={n+d \choose d}$.}
Let $h_1,\ldots,h_{N_n}$ form a polynomial basis for $\mathrm{Pol}_n(E)$ and denote $H_n(x)=(h_1(x),\ldots,h_{N_n}(x))^\top$. Since $\Gcal$ leaves $\mathrm{Pol}_n(E)$ invariant, there exists a unique matrix $G_n\in\R^{N_n\times N_n}$ representing the action of $\Gcal$ on $\mathrm{Pol}_n(E)$ with respect to the basis $H_n(x)$.
Without loss of generality we assume to work with the monomial basis.
\begin{example}
If $n=1$, then we have $H_1(x)=(1,x_1,\ldots,x_d)^\top$ and $G_1$ becomes
\begin{equation}
G_1=\begin{pmatrix}
0&0\\
\kappa\theta & -\kappa
\end{pmatrix}.
\label{eq:conditional_exp}
\end{equation}
\end{example}
From the invariance property of $\Gcal$, one can derive the moment formula (Theorem 2.4 in  \cite{Filipovic2017})
\begin{align}
\E_t[H_n(X_T)]=\e^{G_n(T-t)}H_n(X_t),
\label{eq:moment_formula}
\end{align}
for all $t\le T$. Many efficient algorithms exist to numerically compute the matrix exponential (e.g., \cite{al2011computing}).

\subsection{Dividend futures}
Consider a stock that pays a continuous dividend stream to its owner at an instantaneous rate $D_t$, which varies stochastically over time. We model the cumulative dividend process $C_t=C_0+\int_0^t D_s\,\dd s$ as:
\begin{equation}
C_t=\e^{\beta t} p^\top H_1(X_t),
\label{eq:cumul_dividend_spec}
\end{equation}
for some parameters $\beta\in\R$ and $p\in\R^{d+1}$ such that $C_t$ is a positive, non-decreasing, and absolutely continuous (i.e., drift only) process.
This specification for $C_t$ implicitly pins down $D_t$, which is shown in the following proposition.
\begin{proposition}\label{prop:dividend_rate} The instantaneous dividend rate $D_t$ implied by \eqref{eq:cumul_dividend_spec} is given by
\begin{equation}
D_t=\e^{\beta t}p^\top (\beta \mathrm{Id}+G_1)H_1(X_t),
\label{eq:dividend_spec}
\end{equation}
where $\mathrm{Id}$ denotes the identity matrix.
\end{proposition}
Remark that both the instantaneous dividend rate and the cumulative dividends load linearly on the factor process. The exponential scaling of $C_t$ with parameter $\beta$ can be helpful to guarantee a non-negative instantaneous dividend rate. Indeed, if
\begin{equation}\lambda=\displaystyle \sup_{x\in E}-\frac{p^\top G_1 H_1(x)}{p^\top H_1(x)}
\label{eq:lambda}
\end{equation}
is finite, then it follows from \eqref{eq:dividend_spec} that $D_t\ge 0$ if and only if $\beta \ge \lambda$.\footnote{We calculate $\lambda$ explicitly for the  linear jump-diffusion model studied in Section \ref{section:LJD}.}
Moreover, when all eigenvalues of $\kappa$ have positive real parts, it follows from the moment formula \eqref{eq:moment_formula} that
\[
\lim_{T\to\infty}\frac{1}{T-t}\log \left(\frac{\E_t[D_T]}{D_t}\right)=\beta.
\]
The parameter $\beta$ therefore controls the asymptotic risk-neutral expected growth rate of the dividends.

The time-$t$ price of a continuously marked-to-market futures contract referencing the dividends to be paid over a future time interval $[T_1,T_2]$ with expiry date $T_2$, $t\le T_1\le T_2$, is given by:
\begin{align}
D_{fut}(t,T_1,T_2)
&=\E_t\left[\int_{T_1}^{T_2}D_s\,\dd s\right]\nonumber \\
&=\E_t\left[C_{T_2}-C_{T_1}\right]\nonumber \\
&=p^\top
\left(\e^{\beta T_2}\e^{G_1(T_2-t)}-\e^{\beta T_1}\e^{G_1(T_1-t)}\right)
H_1(X_t), \label{eq:price_div_futures}
\end{align}
where we have used the moment formula \eqref{eq:moment_formula} in the last equality.
Hence, the dividend futures price is linear in the factor process. Note that the dividend futures term structure (i.e., the dividend futures prices for varying $T_1$ and $T_2$) does not depend on the specification of the martingale part of $X_t$.

\subsection{Bonds and swaps}
Denote the risk-neutral discount factor by $\zeta_t$. It is related to the short rate $r_t$ as follows
\[\zeta_T=\zeta_t \e^{-\int_t^Tr_s\,\dd s},\quad 0\le t\le T.\]
We directly specify dynamics for the risk-neutral discount factor:
\begin{align}
\zeta_t=\e^{-\gamma t} q^\top H_1(X_t),\label{eq:discount_factor}
\end{align}
for some parameters $\gamma\in\R$ and $q\in \R^{d+1}$ such that $\zeta_t$ is a  positive and absolutely continuous process. This is similar to the specification \eqref{eq:cumul_dividend_spec} of $C_t$ but, in order to allow for negative interest rates, we do not require $\zeta_t$ to be monotonic (non-increasing). \cite{filipovic2014linear} follow a similar approach and specify linear dynamics for the state price density with respect to the historical probability measure $\Pa$. Their specification pins down the market price of risk. It turns out that the polynomial property of the factor process is not preserved under the change of measure from $\Pa$ to $\Q$ in this case. However, as seen in \eqref{eq:price_div_futures}, the polynomial property (in particular the linear drift) under $\Q$ is important for pricing the dividend futures contracts.

The time-$t$ price of a zero-coupon bond paying one unit of currency at time $T\ge t$ is given by:
\[
P(t,T)=\frac{1}{\zeta_t}\E_t\left[\zeta_T\right].
\]
Using the moment formula \eqref{eq:moment_formula} we get a linear-rational expression for the zero-coupon bond price
\begin{equation}
P(t,T)=\e^{-\gamma(T-t)}\frac{q^\top\,
\e^{G_1(T-t)}H_1(X_t)}{q^\top H_1(X_t)}.
\label{eq:price_bond}
\end{equation}
Remark that the term structure of zero-coupon bond prices depends only on the drift of $X_t$. Similarly as in \cite{filipovic2014linear}, one can introduce exogenous factors feeding into  the martingale part of $X_t$ to generate unspanned stochastic volatility (see e.g., \cite{collin2002bonds}), however we do not consider this in our paper.

Using the relation $r_t=-\partial_T \log P(t,T)\vert_{T=t}$, we obtain the following linear-rational expression for the short rate:
\begin{equation*}
r_t=\gamma-\frac{q^\top G_1 H_1(X_t)}{q^\top H_1(X_t)}.
\end{equation*}
When all eigenvalues of $\kappa$ have positive real parts, it follows that
\[
\lim_{T\to\infty}-\frac{\log (P(t,T))}{T-t}=\gamma,
\]
so that $\gamma$ can be interpreted as the yield on the zero-coupon bond with infinite maturity.

Ignoring differences in liquidity and credit characteristics between discount rates and IBOR rates, we can value swap contracts as linear combinations of zero-coupon bond prices. The time-$t$ value of a payer interest rate swap with first reset date $T_0\ge t$, fixed leg payment dates $T_1<\cdots<T_n$, and fixed rate $K$ is given by:
\begin{equation}
\pi_t^{swap}=P(t,T_0)-P(t,T_n)- K\sum_{k=1}^n\delta_k P(t,T_k),
\label{eq:swap_value}
\end{equation}
with $\delta_k=T_k-T_{k-1}$, $k=1\ldots,n$. The forward swap rate is defined as the fixed rate $K$ which makes the right hand side of \eqref{eq:swap_value} equal to zero. Note that the discounted swap value $\zeta_t \pi_t^{swap}$ becomes a linear function of $X_t$, which will be important for the purpose of pricing swaptions.

\subsection{Dividend paying stock}
Denote by $S^\ast_t$ the \emph{fundamental price} of the stock, which we define as the present value of all future dividends:
\begin{equation}
S^\ast_t=\frac{1}{\zeta_t}\E_t\left[\int_t^\infty \zeta_sD_s\,\dd s\right].
\label{eq:fundamental_stock}
\end{equation}
In order for $S_t^\ast$ to be finite in our model, we must impose parameter restrictions. The following proposition provides sufficient conditions on the parameters, together with a closed form expression for $S_t^\ast$. The latter is derived using the fact that $\zeta_t D_t$ is quadratic in $X_t$, hence we are able to calculate its conditional expectation through the moment formula \eqref{eq:moment_formula}.
\begin{proposition}\label{prop:stock_price}
If the real parts of the eigenvalues of $G_2$ are bounded above by $\gamma-\beta$, then $S_t^\ast$ is finite and given by
\begin{equation}
S_t^\ast=\e^{\beta t}\frac{w^\top\,H_2(X_t)}{q^\top H_1(X_t)},
\label{eq:stock_price}
\end{equation}
where $w=\left[(\gamma-\beta)\, \mathrm{Id}-G_2^{\top}\right]^{-1}v$ and $v\in \R^{N_2}$ is the unique coordinate vector satisfying
\[
v^\top H_2(x)= p^\top (\beta \mathrm{Id}+G_1) H_1(x)\, q^\top H_1(x).
\]
\end{proposition}
Proposition \ref{prop:stock_price} shows that the discounted fundamental stock price $\zeta_t S_t^\ast$ is quadratic in $X_t$, which means in particular that we have all  moments of $\zeta_t S_t^\ast$ in closed form. Loosely speaking, the fundamental stock price will be finite as long as the dividends are discounted at a sufficiently high rate (by choosing $\gamma$ sufficiently large). Henceforth we will assume that the assumption of Proposition \ref{prop:stock_price} is satisfied.

The following proposition shows how the price of the dividend paying stock, which we denote by $S_t$, is related to the fundamental stock price.\footnote{This relationship has been highlighted in particular by \cite{buehler2010volatility,buehler2015volatility} in the context of derivative pricing.}
\begin{proposition}\label{prop:stock_price_formal}
The market is arbitrage free if and only if $S_t$ is of the form
\begin{align}
S_t=S_t^\ast+\frac{L_t}{\zeta_t},\label{eq:stock_price_formal}
\end{align}
with $L_t$ a non-negative local martingale.
\end{proposition}

The process $L_t$ can be interpreted as a bubble in the sense that it drives a wedge between the fundamental stock price and the observed stock price. If $X_t$ is continuous, then applying It\^o's lemma to \eqref{eq:stock_price_formal} and using the fact that $\zeta_t$ is assumed to be absolutely continuous, we obtain the following risk-neutral stock price dynamics
\begin{equation}
\dd S_t=(r_tS_t-D_t)\,\dd t+\e^{\beta t}\frac{w^\top\,\mathrm{J}_{H_2}(X_t)}{q^\top H_1(X_t)}
\,\dd M_t +\frac{1}{\zeta_t}\dd L_t,
\label{eq:stock_dynamics}
\end{equation}
where $\mathrm{J}_{H_2}(x)$ denotes the Jacobian of $H_2(x)$.\footnote{A similar, but lengthier, expression can be derived in case there are jumps in $X_t$. We choose to omit it since it does not add much value to the discussion that follows.}    Remark that $S_t$ has the correct risk-neutral drift, by construction. Given dynamics for $r_t$ and $D_t$, an alternative approach to model $S_t$ for derivative pricing purposes would have been to directly specify its martingale part. With such an approach, however, it is not straightforward to guarantee a positive stock price. Indeed, the downward drift of the instantaneous dividend rate could push the stock price in negative territory.\footnote{Instead of starting from dynamics for $D_t$, we could have specified dynamics for the dividend yield $D_t/S_t$. This would help to keep the stock price positive, but it does typically not produce a tractable distribution for $D_t$. This is problematic since dividend derivatives reference notional dividend payments paid out over a certain time period.}
Moreover, by directly specifying the martingale part of the stock price, we are implicitly modeling a bubble because the stock price will be greater than the present value of all future dividends in general. In contrast, our approach implies a martingale part (the second term in \eqref{eq:stock_dynamics}) that guarantees a positive stock price. This martingale part is completely determined by the given specification for dividends and interest rates. In case this is too restrictive for the stock price dynamics, one can always adjust accordingly through the specification of the non-negative local martingale $L_t$. For example, \cite{buehler2015volatility} considers a local volatility model on top of the fundamental stock price that is separately calibrated to equity option prices.

\begin{remark}
Bubbles are usually associated with strict local martingales, see e.g.\ \cite{cox2005local}. In fact, for economies with a finite time horizon, a bubble is only possible if the deflated gains process is a strict local martingale, which corresponds to a bubble of Type 3 according to the classification of \cite{jarrow2007asset}. For economies with an infinite time horizon, which is the case in our setup, bubbles are possible even if the deflated gains process is a true martingale. Such bubbles are of Type 1 and 2 in the classification \cite{jarrow2007asset}. Specifically, a (uniformly integrable) martingale $L_t$ corresponds to a bubble of Type 2 (Type 1).
\end{remark}

The $T$-forward stock price at time $t\le T$ is defined as
\[
F(t,T) =\frac{1}{\zeta_t} \frac{\E_t[\zeta_{T} S_{T}]}{P(t,T)} .
\]
If $L_t$ is a true martingale, then we can compute $F(t,T)$ explicitly in our framework using the moment formula \eqref{eq:moment_formula}
\begin{align*}
F(t,T) & = \frac{\E_t[\zeta_{T} S_{T}]}{\E_t[\zeta_{T}]}  
= \frac{\e^{(\beta - \gamma) T} w^\top \E_t[H_2(X_T)] + \E_t[L_T]}{\e^{- \gamma T}q^\top \E_t[H_1(X_T)]}
= \frac{\e^{(\beta - \gamma) T} w^\top \e^{G_2(T-t)}H_2(X_t) + L_t}{\e^{- \gamma T}q^\top \e^{G_1(T-t)}H_1(X_t)}.
\end{align*}

We finish this section with a result on the duration of the stock. We define the stock duration as
\begin{equation}\label{eq:duration}
Dur_t=\frac{\int_t^\infty (s-t)\,\E_t[\zeta_s D_s]\,\dd s}{\zeta_t S_t^\ast}.
\end{equation}
The stock duration represents a weighted average of the time an investor has to wait to receive his dividends, where the weights are the relative contribution of the present value of the dividends to the fundamental stock price. This definition is the continuous time version of the one used by \cite{dechow2004implied} and \cite{weber2018cash}. The following proposition gives a closed form expression for stock duration in our framework.
\begin{proposition}\label{prop:duration}
The stock duration is given by
\begin{equation}
Dur_t=\frac{w^\top\left[(\gamma-\beta)\, \mathrm{Id}-G_2\right]^{-1}H_2(X_t)}{w^\top H_2(X_t)}.
\end{equation}
\end{proposition}

\section{Option pricing}\label{section:option_pricing}
In this section we address the problem of pricing derivatives with discounted payoff functions that are not polynomials in the factor process. The polynomial framework no longer allows to price such derivatives in closed form. However, we can accurately approximate the prices using the available moments of the factor process.

\subsection{Maximum entropy moment matching}\label{section:ME}
In all examples encountered below, we consider a derivative maturing at time $T$ whose discounted payoff is given by $F(g(X_T))$, for some $g\in \mathrm{Pol}_n(E)$, $n\in\N$, and some function $F\colon\R\to\R$. The time-$t$ price $\pi_t$ of this derivative is given by
\begin{equation}
\pi_t=\E_t\left[F\big(g(X_T\big)\right].
\label{eq:scalar_derivative_payoff}
\end{equation}
If the conditional distribution of the random variable $g(X_T)$ were available in closed form, we could compute $\pi_t$ by integrating $F$ over the real line. In general, however, we are only given all the conditional moments of the random variable $g(X_T)$. We thus aim to construct an approximative probability density function $f$ matching a finite number of these moments. In a second step we approximate the option price through numerically integrating $F$ with respect to $f$. Given that a function is an infinite dimensional object, finding such a function $f$ is clearly an underdetermined problem and we need to introduce additional criteria to pin down one particular function. A popular choice in the engineering and physics literature is to choose the density function with maximum entropy:
\begin{equation}
\def\arraystretch{2}
\begin{array}{ccc}
 \displaystyle \max_{f}&-\displaystyle\int_R f(x)\ln f(x)\,\dd x&\\
 \mathrm{s.t.}&\displaystyle\int_R x^nf(x)\,\dd x=M_n,&\quad n=0,\ldots,N,
\end{array}
\label{eq:entropy_optim}
\end{equation}
where $R\subseteq \R$ denotes the support and $M_0=1,M_1,\ldots,M_N$ denote the first $N+1$ moments of $g(X_T)$. \cite{jaynes1957information} motivates such a choice by noting that maximizing entropy incorporates the least amount of prior information in the distribution, other than the imposed moment constraints. In this sense it is maximally noncommittal with respect to unknown information about the distribution.

Straightforward functional variation with respect to $f$ gives the following solution to this optimization problem:
\begin{align*}
f(x)=\exp\left(-\sum_{i=0}^N\lambda_i x^i\right),\quad x\in R,
\end{align*}
where the Lagrange multipliers $\lambda_0,\ldots,\lambda_N$ have to be solved from the moment constraints:
\begin{align}
\int_R x^n\exp\left(-\sum_{i=0}^N\lambda_i x^i\right)\,\dd x=M_n, \quad n=0,\ldots,N.
\label{eq:ME_lagrange}
\end{align}

If $N=0$ and $R=[0,1]$, then we recover the uniform distribution. For $N=1$ and $R=(0,\infty)$ we obtain the exponential distribution, while for $N=2$ and $R=\R$ we obtain the Gaussian distribution. For $N\ge 3$, one needs to solve the system in \eqref{eq:ME_lagrange} numerically, which involves evaluating the integrals numerically.\footnote{
Directly trying to find the roots of this system might not lead to satisfactory results. A more stable numerical procedure is obtained by introducing the following potential function:
$
P(\lambda_0,\ldots,\lambda_N)=\int_R \exp(-\sum_{i=0}^N\lambda_i x^i)\,\dd x + \sum_{i=0}^N\lambda_i M_i
$.
This function can easily be shown to be everywhere convex (see e.g., \cite{mead1984maximum}) and its gradient corresponds to the vector of moment conditions in \eqref{eq:ME_lagrange}. In other words, the Lagrange multipliers can be found by minimizing the potential function $P(\lambda_0,\ldots,\lambda_N)$. This is an unconstrained convex optimization problem where we have closed form (up to numerical integration) expressions for the gradient and hessian, which makes it a prototype problem to be solved with Newton's method. }
We refer to the existing literature for more details on the implementation of maximum entropy densities, see e.g. \cite{agmon1979algorithm}, \cite{mead1984maximum}, \cite{rockinger2002entropy}, and  \cite{holly2011fourth}.

\begin{remark}\label{remark:path_dependent}
By subsequently combining the law of iterated expectations and the moment formula \eqref{eq:moment_formula}, we are also able to compute the conditional moments of the finite dimensional distributions of $X_t$. In particular, the method described in this section can also be applied to price path-dependent derivatives whose discounted payoff depends on the factor process at a finite number of future time points. One example of such products are the dividend options, which will be discussed below.
\end{remark}

\subsection{Swaptions, stock and dividend options}
The time-$t$ price $\pi^{swaption}_t$ of a payer swaption with expiry date $T_0$, which gives the owner the right to enter into a (spot starting) payer swap at $T_0$, is given by:
\begin{align*}
\pi^{swaption}_t&=\frac{1}{\zeta_t}\E_t\left[\zeta_{T_0}\left(\pi^{swap}_{T_0}\right)^+\right]\\
&=\frac{1}{\zeta_t}\E_t\left[\left(\zeta_{T_0}-\zeta_{T_0}P(T_0,T_n)- K\sum_{k=1}^n\delta_k \zeta_{T_0}P(T_0,T_k)\right)^+\right]\\
&=\frac{\e^{-\gamma(T_0-t)}}{\vphantom{\big|} q^\top H_1(X_t)}\E_t\left[\left(q^\top\left(\mathrm{Id}-\e^{(G_1-\gamma\mathrm{Id})(T_n-T_0)}- K\sum_{k=1}^n\delta_k\e^{(G_1-\gamma\mathrm{Id})(T_k-T_0)}\right)H_1(X_{T_0})\right)^+\right],
\end{align*}
where we have used \eqref{eq:price_bond} in the last equality.
Observe that the discounted payoff of the swaption is of the form in \eqref{eq:scalar_derivative_payoff} with $F(\cdot)=\max(\cdot,0)$ and $g$ is a polynomial of degree one in $X_{T_0}$.

The time-$t$ price $\pi_t^{stock}$ of a European call option on the dividend paying stock with strike $K$ and expiry date $T$ is given by
\begin{align}
\pi_t^{stock}&=\frac{1}{\zeta_t}\E_t\left[\zeta_T(S_T-K)^+\right]\nonumber \\
&=\frac{1}{\zeta_t}\E_t\left[(L_T+\zeta_TS_T^\ast-\zeta_T K)^+\right]\nonumber \\
&=\frac{ \e^{-\gamma (T-t)}}{\vphantom{\big|} q^\top H_1(X_t)}\E_t\left[\left(\e^{\gamma T}L_T+\e^{\beta T}w^\top H_2(X_T)-q^\top H_1(X_T) K\right)^+\right],\label{eq:stock_option}
\end{align}
where we have used \eqref{eq:stock_price} in the last equality.
If $(L_t,X_t)$ is jointly a polynomial jump-diffusion, we can compute all moments of the random variable $\e^{\gamma T}L_T+\e^{\beta T}w^\top H_2(X_T)-q^\top H_1(X_T) K$ and proceed as explained in Section \ref{section:ME}.
\begin{remark}
If one assumes independence between the processes $L_t$ and $X_t$, then the assumption that $(L_t,X_t)$ must jointly be a polynomial jump-diffusion is not necessarily needed. Indeed, suppose $L_t$ is specified such that we can compute $F(k)=\e^{-\gamma (T-t)}\E_t[(\e^{\gamma T}L_T-k)^+]$ efficiently. By the law of iterated expectations we have
\begin{align*}
\pi_t^{stock}=\frac{\E_t\left[F(g(X_T))\right]}{\vphantom{\big|} q^\top H_1(X_t)},
\end{align*}
where we define $g(x)=-\e^{\beta T}w^\top H_2(x)+q^\top H_1(x) K\in \mathrm{Pol}_2(E)$. The numerator in the above expression is now of the form in \eqref{eq:scalar_derivative_payoff} and we proceed as before.
\end{remark}

Consider next a European call option on the dividends realized in $[T_1,T_2]$, expiry date $T_2$, and strike price $K$. This type of options are actively traded on the Eurex exchange where the Euro Stoxx 50 dividends serve as underlying. The time-$t$ price $\pi^{div}_t$ of this product is given by
\begin{align*}
\pi^{div}_t&=\frac{1}{\zeta_{t}}\E_t\left[\zeta_{T_2}\left(\int^{T_2}_{T_1}D_s\,\dd s-K\right)^+\right]\\
&=\frac{1}{\zeta_{t}}\E_t\left[\left(\zeta_{T_2}(C_{T_2}-C_{T_1}-K)\right)^+\right]\\
&=\frac{ \e^{-\gamma (T_2-t)}}{\vphantom{\big|} q^\top H_1(X_t)}
\E_t\left[\left(q^\top H_1(X_{T_1})\left(\e^{\beta T_2}p^\top H_1(X_{T_2})-\e^{\beta T_1}p^\top H_1(X_{T_1})-K\right)\right)^+\right].
\end{align*}
We can compute in closed form all the moments of the scalar random variable
\[q^\top H_1(X_{T_2})\left(\e^{\beta T_2}p^\top H_1(X_{T_2})-\e^{\beta T_1}p^\top H_1(X_{T_1})-K\right)\]
by subsequently applying the law of iterated expectations and the moment formula \eqref{eq:moment_formula}, see Remark \ref{remark:path_dependent}. We then proceed as before by finding the maximum entropy density corresponding to these moments and computing the option price by numerical integration.

\subsection{Interest rate-dividend hybrid option}\label{sec:hybrid}
We describe in this section an interest rate-dividend hybrid derivative that gives direct exposure to dividend payments and interest rate movements. Consider a tenor structure $T_0<\cdots <T_N$. At time $T_k$, $k=1,\ldots, N$, the derivative pays the positive part of the difference between the dividends realized over $[T_{k-1},T_k]$, normalized by the $T_{k-1}$-forward stock price, and the in-arrears compounded risk-free rate augmented with a constant spread $s\in\R$:
\[
\left(\frac{1}{F(T_0,T_{k-1})}\int_{{T_{k-1}}}^{T_k} D_u\,\dd u 
-(T_k-T_{k-1})(R^c(T_{k-1},T_k) + s\big)\right)^+,
\] 
where we define the in-arrears compounded risk-free rate as
\[
R^c(T_{k-1},T_k) = \frac{1}{T_k-T_{k-1}}\left(\e^{\int_{T_{k-1}}^{T_k} r_u \,\dd u}-1\right).
\]
%
%
This payoff structure is particularly relevant in the context of the transition of LIBOR to alternative risk-free rates (ARFRs), where term rates are constructed by compounding daily fixings of a benchmark rate based on overnight rates. In our setting, we proxy the overnight rate by the short rate and the daily compounding by continuous compounding. 

The price $\pi_{T_0}^{hybrid}$ at time $T_0$ is given by
\begin{align*}
\pi_{T_0}^{hybrid}&=\frac{1}{\zeta_{T_0}}\E_{T_0}\left[\sum_{k=1}^N \zeta_{T_k}\left(\frac{1}{F(T_0,T_{k-1})}(C_{T_k} - C_{T_{k-1}}) -  \left(\frac{\zeta_{T_{k-1}}}{\zeta_{T_k}} -1\right) - s(T_k-T_{k-1})\right)^+ \right]\\
&=\frac{1}{\zeta_{T_0}}\sum_{k=1}^N \E_{T_0}\left[\left(\frac{1}{F(T_0,T_{k-1})}\zeta_{T_k}(C_{T_k} - C_{T_{k-1}}) -  \left(\zeta_{T_{k-1}} -\zeta_{T_k}\right) - s\zeta_{T_k}(T_k-T_{k-1})\right)^+ \right].
\end{align*}
We can compute in closed form all the $\Fcal_{T_0}$-conditional moments of the scalar random variables
\[
\frac{1}{F(T_0,T_{k-1})}\zeta_{T_k}(C_{T_k} - C_{T_{k-1}}) -  \left(\zeta_{T_{k-1}} -\zeta_{T_k}\right) - s\zeta_{T_k}(T_k-T_{k-1}),\quad k=1,\ldots,N,
\]
by subsequently applying the law of iterated expectations and the moment formula \eqref{eq:moment_formula}, see Remark \ref{remark:path_dependent}. We then proceed as before by finding the maximum entropy density corresponding to these moments and computing the option price by numerical integration.

\section{The linear jump-diffusion model}\label{section:LJD}
In this section we give a worked-out example of a factor process that fits in the polynomial framework of Section \ref{section:pol_FW}. In the following, if $x\in\R^d$ then $\mathrm{diag}(x)$ denotes the diagonal matrix with $x_1,\ldots,x_d$ on its diagonal. If $x\in \R^{d\times d}$, then we denote $\mathrm{diag}(x)=(x_{11},\ldots,x_{dd})^\top$.

The \textit{linear jump-diffusion} (LJD) model assumes the following dynamics for the factor process
\begin{align}
\dd X_t=\kappa(\theta-X_{t})\,\dd t +\mathrm{diag}(X_{t-})\left( \Sigma\,\dd B_t + \dd J_t\right),
\label{eq:LJD}
\end{align}
where $B_t$ is a standard $d$-dimensional Brownian motion, $\Sigma\in\R^{d\times d}$ is a lower triangular matrix with non-negative entries on its main diagonal, $J_t$ is a compensated compound Poisson process with arrival intensity $\xi\ge 0$ and a jump distribution $F(\dd z)$ that admits moments of all orders.\footnote{For simplicity we assume a compound Poisson process with a single jump intensity, however this can be generalized (see \cite{Filipovic2017}).}
Both the jump amplitudes and the Poisson jumps are assumed to be independent from the diffusive noise. The purely diffusive LJD specification (i.e., $\xi=0$) has appeared in various financial contexts such as stochastic volatility (\cite{nelson1990arch}, \cite{barone2005option}), energy markets (\cite{Pilipovic1997energy}), interest rates (\cite{brennan1979continuous}), and Asian option pricing (\cite{linetsky2004spectral}, \cite{willems2018asian}). The extension with jumps has not received much attention yet.
%

The following proposition verifies that $X_t$ is indeed a polynomial jump-diffusion and also shows how to choose parameters such that $X_t$ has positive components.
\begin{proposition}\label{prop:posSol}
Assume that matrix $\kappa$ has non-positive off-diagonal elements, $(\kappa\theta)_i\ge 0$, $i=1,\ldots,d$, and $F$ has support $\Scal \subseteq (-1,\infty)^d$. Then for every initial value $X_0\in (0,\infty)^{d}$ there exists a unique strong solution $X_t$ to \eqref{eq:LJD} with values in $(0,\infty)^{d}$. Moreover, $X_t$ is a polynomial jump-diffusion.
\end{proposition}

We will henceforth assume that the assumptions of Proposition \ref{prop:posSol} are satisfied, as it allows to derive parameter restrictions to guarantee  $C_t>0$, $\zeta_t>0$, and $D_t\ge 0$. In order to have $p^\top H_1(x)>0 $ and $q^\top H_1(x)>0 $ for all $x\in (0,\infty)^d$, the vectors $p$ and $q$ must have non-negative components with at least one component different from zero. The following proposition introduces a lower bound on $\beta$ such that $D_t\ge0$.

\begin{proposition}\label{prop:nonegative_dividend_rate}
Let $p=(p_0,p_1,\ldots,p_d)^\top\in [0,\infty)^{1+d}$ and denote $\tilde{p}=(p_1,\ldots,p_d)^\top$. Assume that at least one of the $p_1,\ldots,p_d$ is non-zero, so that dividends are not deterministic. Without loss of generality we assume $p_1,\ldots,p_k>0$ and $p_{k+1},\ldots,p_d=0$, for some $1\le k\le d$. If we denote by $\kappa_j$ the $j$-th column of $\kappa$, then we have $D_t\ge 0$ if and only if
\begin{equation}
\beta \ge
\begin{cases}
\max \left\{\dfrac{\tilde{p}^\top\kappa_1}{p_1},\ldots,\dfrac{\tilde{p}^\top\kappa_{k}}{p_{k}}\right\}& \text{if }\, p_0=0, \\[10pt]
\max \left\{-\dfrac{\tilde{p}^\top\kappa\theta}{p_0},\dfrac{\tilde{p}^\top\kappa_1}{p_1},\ldots,\dfrac{\tilde{p}^\top\kappa_{k}}{p_{k}}\right\}& \text{if }\, p_0> 0.
\end{cases}
\label{eq:beta_lower_bound}
\end{equation}

\end{proposition}

The LJD model allows a flexible instantaneous correlation structure between the factors through the matrix $\Sigma$. This is in contrast to non-negative affine jump-diffusions, a popular choice in term structure modeling when non-negative factors are required, see, e.g., \cite{duf_fil_sch_03}. Indeed, as soon as one introduces a non-zero instantaneous correlation between the factors of a non-negative affine jump-diffusion, the affine (and polynomial) property is lost. Correlation between factors can be used to incorporate a dependence between the term structures of interest rates and dividends, but also to model a dependence within a single term structure. The LJD model also allows for state-dependent, positive and negative, jump sizes of the factors. This again is in contrast to non-negative affine jump-diffusions.

The following proposition provides the eigenvalues of the corresponding matrix $G_2$ under the assumption of a triangular form for $\kappa$. Combined with Proposition \ref{prop:stock_price}, this gives sufficient conditions to guarantee a finite stock price in the LJD model.
\begin{proposition}\label{prop:LJD_eigenvalues}
 If $\kappa$ is a triangular matrix, then the eigenvalues of the matrix $G_2$ are
\begin{gather*}
0, -\kappa_{11},\dots, -\kappa_{dd}, \\
-\kappa_{ii}-\kappa_{jj}+(\Sigma\Sigma^\top)_{ij}+\xi\,\int_\Scal z_iz_j \, F(\dd z),\quad 1\le i,j\le d.
\end{gather*}
The eigenvalues of $G_1$ coincide with the values on the first line.
\end{proposition}

\section{Numerical study}\label{section:numerical}
In this section we calibrate a parsimonious LJD model specification using daily market data from February to April 2015 obtained from Bloomberg. The purpose of this calibration exercise is to show that a parsimonious model specification is capable of reproducing derivative prices observed in the market. We do not specify the dynamics of the model under the historical probability measure. Hence, we do not study the evolution of risk-premia over time and focus solely on the risk-neutral pricing of derivatives. We leave a study of risk-premia for future research.

\subsection{Data description}
The dividend paying stock in our calibration study is the Euro Stoxx 50, the leading blue-chip stock index in the Eurozone. The index is composed of fifty stocks of sector leading companies from twelve Eurozone countries. We choose to focus on the European market because the dividend futures contracts on the Euro Stoxx 50 are the most liquid in the world and have been around longer than in any other market. \cite{kragt2014dividend} report an average daily turnover of more than EUR 150 million for all expiries combined.
The Euro Stoxx 50 dividend futures contracts are traded on Eurex and reference the sum of the declared ordinary gross cash dividends (or cash-equivalent, e.g.\ stock dividends) on index constituents that go ex-dividend during a given calendar year, divided by the index divisor on the ex-dividend day. Corporate actions that cause a change in the index divisor are excluded from the dividend calculations, e.g.\ special and extraordinary dividends, return of capital, stock splits, etc. On every day of the sample there are ten annual contracts available for trading with maturity dates on the third Friday of December. Specifically, the $k$-th to expire contract, $k=1,\ldots,10$, references the dividends paid between the third Friday of December $2014+k-1$ and the third Friday of December $2014+k$. We interpolate adjacent dividend futures contracts using the approach of \cite{kragt2014dividend} to construct contracts with a constant time to maturity of 1 to 9 years.\footnote{We could also calibrate the model without doing any interpolation of the data. However, in order to make the fitting errors of the sequential calibrations more comparable over time, we choose to interpolate all instruments such that they have a constant time to maturity.} In the calibration we use the contracts with maturities in 1, 2, 3, 4, 5, 7, and 9 years. Figure \ref{fig:market_div_fut} plots the interpolated dividend futures prices with 1, 5, 7, and 9 years to maturity.

Next to the Euro Stoxx 50 dividend futures contracts, there also exist exchange traded options on realized dividends. The maturity dates and the referenced dividends of the options coincide with those of the corresponding futures contracts. At every calibration date, we consider the \cite{black1976pricing} implied volatility of an at-the-money (ATM) dividend option with 2 years to maturity. Since dividend option contracts have fixed maturity dates, we interpolate the implied volatility of the second and third to expire ATM option contract.\footnote{We linearly interpolate the total implied variance $\sigma_{Black}^2 \tau$, where $\sigma_{Black}$ denotes the implied volatility and $\tau$ the maturity of the option.}
Figure \ref{fig:market_div_eq_opt} plots the implied volatilities of the dividend options over time.

The term structure of interest rates is calibrated to European spot-starting swap contracts referencing the six month Euro Interbank Offered Rate (Euribor) with tenors of 1, 2, 3, 4, 5, 7, and 10 years. Figure \ref{fig:market_IRS} plots the par swap rates of swaps with tenors of 1, 5, 7, and 10 years. In addition, we also include ATM swaptions with time to maturity equal to 3 months and underlying swap with tenor 10 years. These are among the most liquid fixed-income instruments in the European market. The swaptions are quoted in terms of normal implied volatility and are plotted in Figure \ref{fig:market_swpt}.

We also consider Euro Stoxx 50 index options with ATM strike and a maturity of 3 months. Their prices are quoted in terms of Black-Scholes implied volatility and plotted in Figure \ref{fig:market_div_eq_opt} together with the dividend options implied volatility. Figure \ref{fig:market_index} plots the Euro Stoxx 50 index level over time.

\subsection{Model specification}
We propose a parsimonious four-factor LJD specification without jumps for $X_t=(X_{0t}^I,X_{1t}^I,X_{0t}^D,X_{1t}^D)^\top$
\begin{align}
\left\{
\def\arraystretch{1.2}
\begin{array}{llll}
\dd X_{0t}^I &=\kappa_0^I \left(X^I_{1t}-X_{0t}^I\right)\,\dd t\\
\dd X^I_{1t}&=\kappa^I_1(\theta^I-X^I_{1t})\,\dd t&+&\sigma^{I}X^I_{1t}\,\dd B_{1t}\\
\dd X_{0t}^D &=\kappa_0^D \left(X^D_{1t}-X_{0t}^D\right)\,\dd t\\
\dd X^D_{1t}&=\kappa^D_1(\theta^D-X^D_{1t})\,\dd t&+&\sigma^D X^D_{1t}\left(\rho\,\dd B_{1t}+\sqrt{1-\rho^2}\,\dd B_{2t}\right)
\end{array}
\right.
,
\label{eq:parsimonious_model}
\end{align}
with $\rho\in[-1,1]$, $\kappa_0^I,\kappa_0^D,\kappa_1^I,\kappa_1^D,\theta^I,\theta^D,\sigma^I,\sigma^D>0$, and $X_0\in (0,\infty)^4$. By Proposition \ref{prop:posSol}, $X_t$ takes values in $(0,\infty)^4$. Since we only include options with ATM strike in the calibration, we choose not to include any jumps in the dynamics in order to keep the number of parameters small. We define the cumulative dividend process as
\[C_t=\e^{\beta t}X_{0t}^D,\]
so that $X_{0t}^D$ and $X_{1t}^D$ are driving the term structure of dividends. The corresponding instantaneous dividend rate becomes
\[
D_t=\e^{\beta t}\left(\left(\beta -\kappa_0^D\right)X_{0t}^D +\kappa_0^D X_{1t}^D\right).
\]
Using Proposition \ref{prop:nonegative_dividend_rate}, we guarantee $D_t\ge 0$ by requiring $\beta \ge \kappa_0^D$. In order to further reduce the number of parameters, we set $\beta=\kappa_0^D$, so that $D_t=\e^{\beta t}\beta X_{1t}^D$ and $X_{0t}^D$ no longer enters in the dynamics of $D_t$. We can thus normalize $C_0=X_{00}^D=1$.

The discount factor process is defined as
\[\zeta_t=\e^{-\gamma t}X_{0t}^I,\]
so that $X_{0t}^I$ and $X_{1t}^I$ are driving the term structure of interest rates. The corresponding short rate becomes
\[
r_t=(\gamma+\kappa_0^I)-\kappa_0^I\frac{X_{1t}^I}{X_{0t}^I},
\]
which is unbounded from below and bounded above by $\gamma+\kappa_0^I$.\footnote{In the more general polynomial framework described in Section \ref{section:pol_FW}, it is possible to lower bound the short rate. For example, one can use compactly supported polynomial processes, similarly as in \cite{ackerer2018linear}.} Dividing $\zeta_t$ by a positive constant does not affect model prices, so for identification purposes we normalize $\theta^I=1$.\footnote{For a constant $k>0$, the dynamics of $(\tilde{X}_{0t}^I,\tilde{X}_{1t}^I)=(kX_{0t}^I,kX_{1t}^I)$ is given by
\[
\left\{
\def\arraystretch{1.2}
\begin{array}{llll}
\dd \tilde{X}_{0t}^I &=\kappa_0^I \left(\tilde{X}^I_{1t}-\tilde{X}_{0t}^I\right)\,\dd t\\
\dd \tilde{X}^I_{1t}&=\kappa^I_1(\tilde{\theta}^I-\tilde{X}^I_{1t})\,\dd t&+&\sigma^{I}\tilde{X}^I_{1t}\,\dd B_{1t}
\end{array}
\right.
,
\]
with $\tilde{\theta}^I=k\theta^I$. The dynamics of $(\tilde{X}^I_{0t},\tilde{X}^I_{1t})$ is therefore of the same form as that of $(X^I_{0t},X^I_{1t})$.}

The matrix $\kappa$ is upper triangular and given by
\[
\kappa=
\begin{pmatrix}
\kappa_0^I&-\kappa_0^I & 0 & 0\\
0& \kappa_1^I & 0 & 0\\
0 & 0 & \kappa_0^D&-\kappa_0^D\\
0 & 0 & 0 &\kappa_1^D
\end{pmatrix}.
\]
The diagonal elements, which coincide with the eigenvalues, of $\kappa$ are all positive by assumption. We can therefore interpret $\gamma$ as the asymptotic zero-coupon bond yield and $\beta$ as the asymptotic risk-neutral expected dividend growth rate. Using Propositions \ref{prop:stock_price} and \ref{prop:LJD_eigenvalues}, we introduce the following constraint on the model parameters in order to guarantee a finite stock price:
\[
\gamma-\beta >\max \left\{0\, ,\, (\sigma^I)^2-2\kappa_1^I\, ,\,(\sigma^D)^2-2\kappa_1^D\, ,\, \sigma^I\sigma^D \rho -\kappa_1^I-\kappa_1^D\right\}.
\]
The parameter $\rho\in[-1,1]$ controls the correlation between interest rates and dividends. Specifically, the instantaneous correlation between the dividend rate and the short rate is given by
\begin{equation}
\frac{\dd [ D,r]_t}{\sqrt{\dd [D,D]_t}\,\sqrt{\dd [r,r]_t}}=-\rho,
\label{eq:corr_D_r}
\end{equation}
where $[\cdot,\cdot]_t$ denotes the quadratic covariation.
The minus sign in front of $\rho$ appears because the Brownian motion $B_{1t}$ drives the discount factor, which is negatively related to the short rate.

We set $L_t\equiv 0$ for parsimony, so that the stock price is equal to the present value of all future dividends, i.e., $S_t\equiv S_t^\ast$.
%

\subsection{Calibration}
We minimize the sum of squared differences between the model and market prices using the Nelder-Mead simplex algorithm. The parameters to be optimized are $\beta$, $\kappa_1^D$, $\theta^D$, $\kappa_0^I$, $\kappa_1^I$, $\gamma$, $\sigma^D$, $\sigma^I$, and $\rho$. We propose an efficient way to filter out the latent factors $X^D_{1t}, X^I_{0t}$, and $X^I_{1t}$ on every day of the sample. For a given set of parameters, the dividend futures price \eqref{eq:price_div_futures} is a linear function of $X^D_{1t}$. We solve for $X^D_{1t}$ through a linear least-squares regression from the dividend futures prices. The discounted swap value $\zeta_t \pi_t^{swap}$ in \eqref{eq:swap_value} is a linear function of the latent factors $X^I_{0t}$ and $X^I_{1t}$.\footnote{Since we are using par swap rates, the value of the swap is equal to zero by definition.} Applying It\^o's lemma to $\zeta_t D_t$, it follows that the discounted stock price $\zeta_tS_t^\ast=\E_t[\int_0^\infty \zeta_s D_s\,\dd s]$ given by \eqref{eq:discount_factor} and \eqref{eq:stock_price} is a linear combination of $X^D_{1t}$, $X^I_{0t}$, $X^I_{1t}$, $X^I_{0t}X^D_{1t}$, and $X^I_{1t}X^D_{1t}$. Since we already solved $X^D_{1t}$ from the dividend futures prices,  $\zeta_t S_t^\ast$ becomes a linear function of $X^I_{0t}$ and $X^I_{1t}$. We solve $X^I_{0t}$ and $X^I_{1t}$ through a weighted linear least-squares regression from the swap rates and the stock price. We assign a relatively large weight to the stock price to make sure it is accurately matched by the model.

Although the option pricing technique described in Seciton \ref{section:ME} works in theory for any finite number of moment constraints, there is a computational cost associated with computing the moments on the one hand, and solving the Lagrange multipliers on the other hand. In the calibration, we use moments up to order four to price swaptions, dividend options, and stock options. The number of moments needed for an accurate option price depends on the specific form of the payoff function and on the model parameters. As an example, Figure \ref{fig:moment_convergence} shows prices of a swaption, dividend option, stock option, and a hybrid option as described in Section \ref{sec:hybrid} for different number of moments matched and a realistic set of parameters. For the hybrid option, we set $T_1-T_0 = 1$, $N=1$, and the spread $s$ such that the option is ATM. The swaption, dividend option, and stock option have the same characteristics as the ones used in the calibration. As a benchmark, we perform a Monte-Carlo simulation of the model. We discretize \eqref{eq:parsimonious_model} at a weekly frequency with a simple Euler scheme and simulate $10^5$ trajectories.\footnote{In addition, we also use the corresponding forward contracts as control variates. This variance reduction technique reduces the variance of the Monte-Carlo estimator approximately by a factor 4.} We observe that using four to five moments produces a price approximation that is very close to the Monte-Carlo benchmark

We calibrate the model consecutively to one month of daily data from February, March, and April 2015. Table \ref{table:abs_error} shows the absolute pricing errors in the second, third, and fifth column, respectively. Considering the relatively small number of parameters, the fit is remarkably good. Dividend futures have a mean absolute relative error less than 1\% with few exceptions. The mean absolute error of the swap rates is in the order of basis points for all tenors and all three months. The model only contains two volatility parameters ($\sigma^I$ and $\sigma^D$), but nonetheless produces a relatively good fit with option prices on average. The Eurostoxx 50 index level is matched almost perfectly, thanks to the relatively large weight in the weighted least-squares regression to filter out the latent factors. The fourth and sixth column of Table \ref{table:abs_error} show out-of-sample pricing errors. Specifically, in the fourth (sixth) column we compute the pricing errors in March (April) using the parameters calibrated on February (March) data. The only degrees of freedom in this out-of-sample exercise are the values of the latent factors, which we filter out as explained before. The loss in pricing accuracy out-of-sample is modest, which speaks for the robustness of the model.

Table \ref{table:params} shows the calibrated parameters. The parameters are comparable for the three calibration months, which is in line with the good out-of-sample performance. The parameter $\gamma$, which is the yield of the zero-coupon bond with infinite maturity, is decreasing in the subsequent calibrations, reflecting the decrease in interest rates over the sample period. The parameter $\beta$, which is the asymptotic risk-neutral expected growth rate of the dividends, is always substantially lower than $\gamma$, as required for the stock price to be finite. The term structure of dividend futures is downward sloping over the entire sample period. This is reflected in the calibration by a small value for $\theta^D$, which is the long-term mean of the process $X^D_{1t}$ driving the dividends. Remarkably, $\rho$ is positive for all three months, close to the upper bound of one. In view of \eqref{eq:corr_D_r}, this indicates a highly negative correlation between interest rates and dividends. This negative correlation is a central ingredient in our model, since it increases the volatility of the stock price relative to the dividends and interest rates. This allows to reconcile the relatively large implied volatility of stock options with the relatively small implied volatility of dividend options. From Figure \ref{fig:market_div_eq_opt} we can see that the difference between the dividend and stock option implied volatility was smaller in March than in February and April. This translates in a smaller calibrated $\rho$ in March compared to February and April. 

In Figure \ref{fig:market_fit} we use the February parameters to plot the model prices together with the market prices over the full sample period. The February to March (March to April) regions of the plots are therefore a visualization of the second and fourth (fourth and sixth) column in Table \ref{table:abs_error}. The goodness of fit deteriorates as we move away from the calibration window, which is to be expected. Note that the model is capable of capturing the level of implied volatilities of the stock and dividend options, but it fails to capture the variation over time. This is caused by the volatility structure of the model, where the relative volatility of the dividend factor $X_{1t}^D$ is constant. Enriching the model specification with more factors can help to address this problem, however we leave this for future research. Figure \ref{fig:state} plots the filtered values of $D_t$ and $r_t$ using the parameters calibrated on February data. The plot looks similar when using the March or April parameters.

Figure \ref{fig:duration} plots the stock duration using the February parameters. The stock duration is quite stable over time with an average around 23 years. \cite{dechow2004implied} and \cite{weber2018cash} construct a stock duration measure based on balance sheet data and find an average duration of approximately 15 and 19 years, respectively, for a large cross-section of stocks. The plot looks similar when using the March or April parameters.
%

Table \ref{tabel:computation_time} contains computation times for calculating option prices. The bulk of the computation times is due to the computation of the moments of $g(X_T)$ in \eqref{eq:scalar_derivative_payoff}. The number of stochastic factors that drive a derivative's payoff and the degree of moments that have to be matched therefore strongly affect the computation time. We observe that all timings of the maximum entropy method are well below the time it took to run the benchmark Monte-Carlo simulation. The pricing of swaptions is much faster than the pricing of dividend and stock options, especially as the number of moments increases. This is because the discounted swaption payoff only depends on on the 2-dimensional process $(X_{0t}^I,X_{1t}^I)^\top$, while the discounted payoff of the dividend and stock option depends on the entire 4-dimensional process $X_t=(X_{0t}^I,X_{1t}^I,X_{0t}^D,X_{1t}^D)^\top$. In addition, the discounted payoff of the dividend and stock option is quadratic in the factors. Therefore, in order to compute moments up to degree $N$ of the discounted payoff, we need to compute moments up to degree $2N$ of the factors. The computation of the dividend option is further complicated by the path-dependent nature of its payoff. Indeed, the dividend option payoff depends on the realization of the factors at $T_1$ \emph{and} $T_2$. In order to compute the moments of $\zeta_{T_2}(C_{T_2}-C_{T_1})$, we have to apply the moment formula twice. Hence, it involves computing a matrix exponential twice, which causes an additional computation time compared to the stock option.

\section{Extensions}\label{section:extensions}
\subsection{Seasonality}\label{section:seasonality}
It is well known that some stock markets exhibit a strongly seasonal pattern in the payment of dividends. For example, Figure \ref{fig:seasonality} shows that the constituents of the Euro Stoxx 50 pay a large fraction of their dividends between April and June each year.\footnote{See e.g.\ \cite{marchioro2016seasonality} for a study of dividend seasonality in other markets.} The easiest way to incorporate seasonality in our framework is to introduce a deterministic function of time $\delta(t)$ and redefine the cumulative dividend process as:
\begin{align}
C_t=\int_0^t\delta(s)\,\dd s + \e^{\beta t}p^\top H_1(X_t).
\label{eq:additive_seasonality}
\end{align}
The function $\delta(t)$ therefore adds a deterministic shift to the instantaneous dividend rate:
\begin{equation}
D_t=\delta(t)+\e^{\beta t}p^\top (\beta \mathrm{Id}+G_1)H_1(X_t).
\end{equation}
In addition to incorporating seasonality, $\delta(t)$ can also be chosen such that the observed dividend futures prices are perfectly matched. In Appendix \ref{section:appendix_bootstrap} we show how the bootstrapping method of \cite{filipovic2016exact} can be used to find such a function.  We do not lose any tractability with the specification in  \eqref{eq:additive_seasonality}, since the moments of $C_{T_2}-C_{T_1}$ can still easily be computed.

Alternatively, we could also introduce time dependence in the specification of $X_t$. Doing so in general comes at the cost of losing tractability, because we leave the class of polynomial jump-diffusions. However, it is possible to introduce a specific type of time dependence such that we \emph{do} stay in the class of polynomial jump-diffusions. Define $\Gamma(t)$ as a vector of sine and cosine functions whose frequencies are integer multiples of $2\pi$ (so that they all have period one)
\begin{equation*}
\Gamma(t)=\begin{pmatrix}
\sin(2\pi t)\\ \cos(2\pi t)\\\vdots \\ \sin(2\pi K t)\\ \cos(2\pi Kt)
\end{pmatrix} \in\R^{2K}, \quad K\in\N, \quad t\ge 0.
\end{equation*}
The superposition
\[z_0+z^\top \Gamma(t), \quad (z_0,z)\in\R^{1+2K},\]
is a flexible function for modeling annually repeating cycles and is a standard choice for pricing commodity derivatives (see e.g.\ \cite{sorensen2002modeling}). In fact, from Fourier analysis we know that any smooth periodic function can be expressed as a sum of sine and cosine waves. Remark now that $\Gamma(t)$ is the solution of the following linear ordinary differential equation
\[
\dd \Gamma(t)=
\mathrm{blkdiag}\left(
\begin{pmatrix}
0&2\pi\\
-2\pi& 0
\end{pmatrix},
\ldots,
\begin{pmatrix}
0&2\pi K\\
-2\pi K& 0
\end{pmatrix}
\right)
\Gamma(t)
\dd t.
\]
The function $\Gamma(t)$ can therefore be seen as a (deterministic) process of the form in \eqref{eq:factor_process} and can be added to the factor process.  For example, the specification for $(X_{0t}^D,X_{1t}^D)$ in \eqref{eq:parsimonious_model} could be replaced by
\[
\left\{
\def\arraystretch{1.2}
\begin{array}{llll}
\dd X_{0t}^D &=\kappa_0^D \left(X^D_{1t}-X_{0t}^D\right)\,\dd t\\
\dd X^D_{1t}&=\kappa^D_1(z_0+z^\top \Gamma(t)-X^D_{1t})\,\dd t&+&\sigma^D X^D_{1t}\left(\rho\,\dd B_{1t}+\sqrt{1-\rho^2}\,\dd B_{2t}\right)
\end{array}
\right.
,
\]
where the first factor mean-reverts around the second, and the second mean-reverts around a time-dependent mean. The process $X_t$ does not belong to the class of polynomial jump-diffusions, however the augmented process $(\Gamma(t),X_t)$ does.

In the calibration exercise in Section \ref{section:numerical}, we did not include any seasonal behavior in the dividends because the instruments used in the estimation are not directly affected by seasonality. Indeed, all the dividend derivatives used in the calibration reference the total amount of dividends paid in a full calendar year. The timing of the dividend payments within the year does therefore not play any role. In theory, the stock price should inherit the seasonal pattern from the dividend payments, since it drops by exactly the amount of dividends paid out. In practice, however, these price drops are obscured by the volatility of the stock price since the dividend payments typically represent only a small fraction of the total stock price. Dividend seasonality only plays a role for pricing claims on dividends realized over a time period different from an integer number of calendar years.

\subsection{Dividend forwards}
Dividend forwards, also known as dividend swaps, are the OTC equivalent of the exchange traded dividend futures. The buyer of a dividend forward receives at a future date $T_2$ the dividends realized over a certain time period $[T_1,T_2]$ against a fixed payment. Dividend forwards differ from dividend futures because they are not marked to market on a daily basis. The dividend forward price $D_{fwd}(t,T_1,T_2)$, $t\le T_1\le T_2$, is defined as the fixed payment that makes the forward have zero initial value
\begin{align*}
D_{fwd}(t,T_1,T_2)&=\frac{1}{P(t,T_2)}\frac{1}{\zeta_t}\E_t\left[\zeta_{T_2}(C_{T_2}-C_{T_1})\right]\\
&=D_{fut}(t,T_1,T_2)+\frac{\mathrm{Cov}_t\left[\zeta_{T_2},C_{T_2}-C_{T_1}\right]}{P(t,T_2)\zeta_t}.
\end{align*}
If interest rates and dividends are independent, then we have $D_{fwd}(t,T_1,T_2)=D_{fut}(t,T_1,T_2)$. However, if there is a positive (negative) dependence between interest rates and dividends, then there is a convexity adjustment and the dividend forward price will be smaller (larger) than the dividend futures price.
 The following proposition derives the dividend forward price in the polynomial framework.
\begin{proposition}\label{prop:dividend_swap}
The dividend forward price is given by
\begin{align*}
D_{fwd}(t,T_1,T_2)=\frac{\left(\e^{\beta T_2}w_2^\top\e^{G_2(T_2-t)}-\e^{\beta T_1}w_1^\top \e^{G_2(T_1-t)}\right) H_2(X_t)}{q^\top \e^{G_1(T_2-t)}H_1(X_t)},
\end{align*}
where $w_1,w_2\in\R^{N_2}$ are the unique coordinate vectors satisfying
\begin{align*}
w_1^\top H_2(x)=p^\top H_1(x) q^\top \e^{G_1(T_2-T_1)} H_1(x),\quad w_2^\top H_2(x)=p^\top H_1(x) q^\top H_1(x).
\end{align*}
\end{proposition}

\section{Conclusion}\label{section:conclusion}
We have introduced an integrated framework designed to jointly price the term structures of dividends and interest rates. The uncertainty in the economy is modeled with a multivariate polynomial jump-diffusion. The model is tractable because we can calculate all conditional moments of the factor process in closed form. In particular, we have derived closed form formulas for prices of bonds, dividend futures, and the dividend paying stock. Option prices are obtained by integrating the discounted payoff function with respect to a moment matched density function that maximizes the Boltzmann-Shannon entropy. We have introduced the LJD model, characterized by a martingale part that loads linearly on the factors. The LJD model allows for a flexible dependence structure between the factors, which offers a valuable alternative to non-negative affine jump-diffusion models. We have assumed that dividends are paid out continuously and ignored the possibility of default. These assumptions are justified when considering derivatives on a stock index, but become questionable for derivatives on a single stock. An interesting future research direction is therefore to extend our framework with discrete dividend payments and default risk. 

\clearpage
\appendix

\section{Bootstrapping an additive seasonality function}
\label{section:appendix_bootstrap}
In this section we explain how to bootstrap a smooth curve $T\mapsto f_t(T)$ of (unobserved) futures prices corresponding to the instantaneous dividend rate $D_T$. The curve should perfectly reproduce observed dividend futures prices and in addition incorporate a seasonality effect. Once we have this function, we define the function $\delta(T)$ as
\[
\delta(T)=f_t(T)-p^\top (\beta \mathrm{Id}+G_1)\E_t[H_1(X_T)],\quad T\ge t,
\]
so that the specification in \eqref{eq:additive_seasonality} perfectly reproduces observed futures contracts and incorporates seasonality.

Suppose for notational simplicity that today is time $0$ and we observe the futures prices $F_i$ of the dividends realized over one calendar year $[i-1,i]$, $i=1,\ldots,I$. Divide the calendar year in $J\ge 1$ buckets and assign a seasonal weight $w_j\ge 0$ to each bucket, with $w_1+\cdots+w_J=1$. These seasonal weights can for example be estimated from a time series of dividend payments. We search for the twice continuously differentiable curve $f_0$ that has maximal smoothness subject to the pricing and seasonality constraints:
\begin{equation*}
\begin{array}{ll}
\displaystyle
\min_{f_0\in C^2(\R)} &\displaystyle f_0(0)^2+f'_0(0)^2+\int^I_0 f''_0(u)^2\,\dd u\\
\multicolumn{1}{c}{\mathrm{s.t.}	}		&\displaystyle \int_{i-1+\frac{j-1}{J}}^{i-1+\frac{j}{J}} f_0(u)\,\dd u=w_jF_i,\quad i=1,\ldots,I,\quad j=1,\ldots,J.				
\end{array}
\end{equation*}
This can be cast in an appropriate Hilbert space as a convex variational optimization problem with linear constraints. In particular, it has a unique solution that can be solved in closed form using similar techniques as presented in \cite{filipovic2016exact}. By discretizing the optimization problem, a non-negativity constraint on $f$ can be added as well.

\section{Proofs}
This section contains all the proofs of the propositions in the paper.

\subsection{Proof of Proposition \ref{prop:dividend_rate}}
Using the moment formula \eqref{eq:moment_formula} we have for $t\le T$
\begin{equation*}
\E_t[C_T]=\e^{\beta T} p^\top \e^{G_1(T-t)}H_1(X_t).
\end{equation*}
Differentiating with respect to $T$ gives
\begin{equation*}
\frac{\dd \E_t[C_T]}{\dd T}=\beta \e^{\beta T} p^\top \e^{G_1(T-t)}H_1(X_t)+\e^{\beta T} p^\top G_1 \e^{G_1(T-t)}H_1(X_t).
\end{equation*}
The result now follows from
\begin{equation*}
D_t=\left.\frac{\dd \E_t[C_T]}{\dd T}\right\rvert_{T=t}.
\end{equation*}
%

\subsection{Proof of Proposition \ref{prop:stock_price}}\label{section:appendix_stock_price}
Plugging in the specifications for $\zeta_t$ and $D_t$ in \eqref{eq:fundamental_stock} gives:
\begin{align*}
S_t^\ast&=\frac{1}{\zeta_t}\int_t^\infty  \e^{-(\gamma-\beta) s}\E_t\left[ p^\top (\beta \mathrm{Id}+G_1) H_1(x)H_1(X_s)\,q^\top H_1(X_s)\right]\,\dd s.
\end{align*}
Since $X_t$ is a polynomial process, we can find a closed form expression for the expectation inside the integral:
\[
\E_t\left[ p^\top (\beta \mathrm{Id}+G_1) H_1(X_s)\, q^\top H_1(X_s)\right]=v^{\top}\e^{G_2(s-t)}\,H_2(X_t).
\]
The fundamental stock price therefore becomes:
\begin{align*}
S^\ast_t&=\frac{\e^{\beta t}v^{\top}}{\vphantom{\big|} q^\top H_1(X_t)}\int_t^\infty \e^{-(\gamma-\beta)(s-t)}\e^{G_2(s-t)}\,\dd s\,H_2(X_t)\\
&=\frac{\e^{\beta t}v^{\top}}{q^\top H_1(X_t)}\left(G_2-(\gamma-\beta) \,\mathrm{Id}\right)^{-1}\exp\left\{\left(G_2-(\gamma-\beta) \,\mathrm{Id}\right)(s-t)\right\}\Bigg\vert^{s=\infty}_{s=t}\,H_2(X_t)\\
&=\frac{\e^{\beta t}v^{\top}}{q^\top H_1(X_t)}\left((\gamma-\beta) \,\mathrm{Id}-G_2\right)^{-1}\,H_2(X_t)\\
&<\infty,
\end{align*}
where we have used the fact that the eigenvalues of the matrix $G_2-(\gamma-\beta)\,\mathrm{Id}$ have negative real parts.

\subsection{Proof of Proposition \ref{prop:stock_price_formal}}
The market is arbitrage free if and only if the deflated gains process
\begin{equation}
G_t=\zeta_tS_t +\int_0^t \zeta_s D_s\,\dd s
\label{eq:gains_process}
\end{equation}
is a non-negative local martingale.

If $S_t$ is of the form in \eqref{eq:stock_price_formal}, then we have
\[
G_t=\E_t\left[\int_0^\infty \zeta_s D_s\,\dd s\right]+L_t,
\]
which is clearly a non-negative local martingale and therefore the market is arbitrage free.

Conversely, suppose that the market is arbitrage free and hence \eqref{eq:gains_process} holds. As a direct consequence, the process
\begin{align*}
\zeta_tS_t-\zeta_tS_t^\ast=G_t-\E_t\left[\int_0^\infty \zeta_s D_s\right]
\end{align*}
must be a local martingale. To show nonnegativity, note that a local martingale bounded from below is a supermartingale, so that we have for all $T\ge t$
\begin{align*}
\zeta_tS_t-\zeta_tS_t^\ast &\ge \E_t\left[G_T -\int_0^\infty \zeta_s D_s\right]\\
&=\E_t\left[\zeta_T S_T -\int_T^\infty \zeta_sD_s\,\dd s\right]\\
&\ge \E_t\left[-\int_T^\infty \zeta_sD_s\,\dd s\right]
\xrightarrow{T\to\infty} 0,
\end{align*}
where we have used the limited liability of the stock in the last inequality.

\subsection{Proof of Proposition \ref{prop:duration}}
Similarly as in the proof of Proposition \ref{prop:stock_price} we get
\begin{align*}
\int_t^\infty (s-t)\,\E_t[\zeta_s D_s]\,\dd s
&=v^\top\int_t^\infty (s-t)\,\e^{(\beta-\gamma)s}\e^{G_2(s-t)} \,\dd s H_2(X_t)\\
&=\e^{(\beta-\gamma)t}v^\top\int_t^\infty (s-t)\,\e^{[G_2-(\gamma-\beta)\mathrm{Id}](s-t)} \,\dd s H_2(X_t).
\end{align*}
Applying integration by parts gives
\begin{align*}
\int_t^\infty (s-t)\,\E_t[\zeta_s D_s]\,\dd s&=
\e^{(\beta-\gamma)t}v^\top [(\gamma-\beta)\mathrm{Id}-G_2]^{-1}
\int_t^\infty \e^{[G_2-(\gamma-\beta)\mathrm{Id}](s-t)} \,\dd sH_2(X_t)\\
&=\e^{(\beta-\gamma)t}v^\top [(\gamma-\beta)\mathrm{Id}-G_2]^{-2}
 H_2(X_t)\\
 &=\e^{(\beta-\gamma)t}w^\top [(\gamma-\beta)\mathrm{Id}-G_2]^{-1}
 H_2(X_t).
\end{align*}
The result now follows from \eqref{eq:stock_price} and \eqref{eq:duration}.

\subsection{Proof of Proposition \ref{prop:posSol}}\label{section:appendix_posSol}
We start by showing that there exists a unique strong solution $X_t$ to \eqref{eq:LJD} with values in $(0,\infty)^{d}$. Due to the global Lipschitz continuity of the coefficients, the SDE in \eqref{eq:LJD} has a unique strong solution in $\R^d$ for every $X_0\in\R^d$, see Theorem III.2.32 in \cite{jacod2003limit}. It remains to show that $X_t$ is $(0,\infty)^d$-valued for all $t\ge 0$ if $X_0 \in (0,\infty)^d$. First, we prove the statement for the diffusive case.
\begin{lemma}
Consider the SDE
\begin{align}
\dd X_t = \kappa(\theta- X_t)\,\dd t+\mathrm{diag}(X_t)\Sigma \,\dd W_t,
\label{eq:posSol_diffusive}
\end{align}
for some $d$-dimensional Brownian motion $W_t$ and $\kappa,\theta,\Sigma$ as assumed in Proposition \ref{prop:posSol}. If $X_0 \in (0,\infty)^d$, then $X_t \in (0,\infty)^d$ for all $t\ge 0$.
\label{lemma:posSol_diffusive}
\end{lemma}
\begin{proof}
Replace $X_t$ in the drift of \eqref{eq:posSol_diffusive} by $X_t^+$ componentwise and consider the SDE
\begin{align}
\dd Y_t = \kappa(\theta - Y_t^+)\,\dd t + \mathrm{diag}(Y_t)\Sigma \,\dd W_t,
\label{eq:posSol_diffusive_cap}
\end{align}
with $Y_0=X_0\in (0,\infty)^d$. The function $y\mapsto y^+$ componentwise is still Lipschitz continuous, so that there exists a unique solution $Y_t$ to \eqref{eq:posSol_diffusive_cap}. Now consider the SDE
\begin{align}
\dd Z_t = -\mathrm{diag}(\mathrm{diag}(\kappa)) Z_t^+ \,\dd t + \mathrm{diag}(Z_t) \Sigma\, \dd W_t,
\label{eq:posSol_GBM}
\end{align}
with $Z_0=X_0\in (0,\infty)^d$. Its unique solution is the $(0,\infty)^d$-valued process given by
\begin{align*}
Z_t=Z_0\exp\left\{\left(-\mathrm{diag}(\kappa)-\frac{1}{2}\mathrm{diag}(\Sigma\Sigma^\top)\right)t+\Sigma\, W_t \right\}.
\end{align*}
By assumption, we have that the drift function of \eqref{eq:posSol_diffusive_cap} is always greater than or equal to the drift function of \eqref{eq:posSol_GBM}:
\[
\kappa\theta -\kappa x^+ \ge-\mathrm{diag}(\mathrm{diag}(\kappa)) x^+,\quad \forall x\in\R^d.\]
By the comparison theorem from \cite[Theorem 1.2]{geiss1994comparison} we have almost surely
\begin{equation*}
Y_t\ge Z_t,\quad t\ge 0.
\end{equation*}
Hence, $Y_t\in (0,\infty)^d$  and therefore $Y_t$ also solves the SDE \eqref{eq:posSol_diffusive}. By uniqueness we conclude that $X_t= Y_t$ for all $t$, which proves the claim.
\end{proof}

Define $\tau_i$ as the $i$th jump time of $N_t$ and $\tau_0=0$. We argue by induction and assume that $X_{\tau_i}>0$ for some $i=0,1,\dots$. Since the process $X_t$ is right-continuous, we have the following diffusive dynamics for the process $X_t^{(\tau_i)}=X_{t+\tau_i}$ on the interval $[0,\tau_{i+1}-\tau_i)$
\begin{equation*}
\dd X_t^{(\tau_i)} =\left(\kappa\theta+\left(-\kappa-\xi \mathrm{diag}\left(\int_\Scal z \,\dd F(\dd z)\right)\right)X_t^{(\tau_i)}\right)\,\dd t+\mathrm{diag}(X_{t}^{(\tau_i)})\Sigma \,\dd B^{(\tau_i)}_{t},
\end{equation*}
with $X_0^{(\tau_i)}=X_{\tau_i}$ and $B^{(\tau_i)}_{t}=B_{\tau_i+t}-B_{\tau_i}$.
The stopping time $\tau_i$ is a.s.\ finite and therefore the process $B^{(\tau_i)}_{t}$ defines a $d$-dimensional Brownian motion with respect to its natural filtration, see Theorem 6.16 in \cite{karatzas1991brownian}. By Lemma \ref{lemma:posSol_diffusive} we have $X_t^{(\tau_i)}\in(0,\infty)^d$ for all $t\in [0,\tau_{i+1}-\tau_i)$. As a consequence, we have $X_t\in(0,\infty)^d$ for all $t\in [\tau_i,\tau_{i+1})$. The jump size $X_{\tau_{i+1}}-X_{\tau_{i+1}-}$ at time $\tau_{i+1}$ satisfies
\begin{equation*}
X_{\tau_{i+1}}-X_{\tau_{i+1}-}=\mathrm{diag}(X_{\tau_{i+1}-})Z_{i+1}>-X_{\tau_{i+1}-},
\end{equation*}
where the $Z_{i+1}$ are i.i.d.\ random variables with distribution $F(\,\dd z)$.
Rearranging terms gives $X_{\tau_{i+1}}\in(0,\infty)^d$. By induction we conclude that $X_t\in(0,\infty)^d$ for $t\in [0,\tau_i)$, $i\in\N$. The claim now follows because $\tau_i \to\infty$ for $i\to\infty$ a.s.

Next, we prove that $X_t$ is a polynomial jump-diffusion. The action of the generator of $X_t$ on a $C^2$ function $f\colon \R^d\to\R$ is given by
\begin{align}
\Gcal f(x)=&\frac{1}{2}\mathrm{tr}\left(\mathrm{diag}(x)\Sigma\Sigma^\top\mathrm{diag}(x)\nabla^2f(x)\right)+\nabla f(x)^\top\kappa(\theta-x)\nonumber\\
&+\xi\left(\int_\Scal f(x+\mathrm{diag}(x)z)\,F(\dd z)- f(x)-\nabla f(x)^\top \mathrm{diag}(x)\int_\Scal z\, F(\dd z)\right),\label{eq:proof_pol_jump_diff}
\end{align}
where $\Scal$ denotes the support of $F$ and we assume that $f$ is such that the integrals are finite. Now suppose that $f\in \mathrm{Pol}_n(\R^d)$ and assume without loss of generality that $f$ is a monomial with $f(x)=x^\alpha=x_1^{\alpha_1}\cdots x_d^{\alpha_d}$, $|\alpha|=n$. We now apply the generator to this function. It follows immediately that the first two terms in \eqref{eq:proof_pol_jump_diff} are again a polynomial of degree $n$ or less. Indeed, the gradient (hessian) in the second (first) term lowers the degree by one (two), while the remaining factors augment the degree by at most one (two). The third term in \eqref{eq:proof_pol_jump_diff} becomes (we slightly abuse the notation $\alpha$ to represent both a multi-index and a vector):
\begin{align}
&\xi \left(x^\alpha \int_\Scal \prod_{j=1}^d (1+z_j)^{\alpha_j}\,F(\dd z)- x^\alpha- x^\alpha\alpha^\top\int_\Scal z\, F(\dd z)\right)\nonumber\\
=&\xi x^\alpha\int_\Scal \left(\e^{\alpha^\top \log(1+z)}- 1- \alpha^\top z\right)\, F(\dd z),\label{eq:appendix_generator_jump}
\end{align}
where the logarithm is applied componentwise.
Hence, we conclude that $\Gcal$ maps polynomials to polynomials of the same degree or less.

\subsection{Proof of Proposition \ref{prop:nonegative_dividend_rate}}
This proof is similar to the one of Theorem 5 in \cite{filipovic2014linear}. From \eqref{eq:dividend_spec} we have that $D_t\ge 0$ if and only if
\begin{equation}
\beta \ge \displaystyle \sup_{x\in (0,\infty)^d}-\frac{p^\top G_1 H_1(x)}{p^\top H_1(x)},
\end{equation}
provided it is finite.
Using \eqref{eq:conditional_exp} we have
\begin{equation}
-\frac{p^\top G_1 H_1(x)}{p^\top H_1(x)}=\frac{-\tilde{p}^\top\kappa\theta+\sum_{j=1}^d \tilde{p}^\top\kappa_j x_j}{p_0+\sum_{j=1}^k p_j x_j}.
\label{eq:lower_bound_div}
\end{equation}
Using the assumption $\kappa_{ij}\le 0$ for $i\neq j$ (cfr., Proposition \ref{prop:posSol}), we have for all $j>k$ that
\begin{equation}
\tilde{p}^\top\kappa_j=\sum_{i=1}^d p_i\kappa_{ij}
=\sum_{i=1}^k p_i\kappa_{ij}\le 0.
\label{eq:lower_bound_div2}
\end{equation}
Combining \eqref{eq:lower_bound_div} with \eqref{eq:lower_bound_div2} gives
\begin{equation}
\sup_{x\in (0,\infty)^d} \frac{-\tilde{p}^\top\kappa\theta+\sum_{j=1}^d \tilde{p}^\top\kappa_j x_j}{p_0+\sum_{j=1}^k p_j x_j}=\sup_{x\in (0,\infty)^k} \frac{-\tilde{p}^\top\kappa\theta+\sum_{j=1}^k \tilde{p}^\top\kappa_j x_j}{p_0+\sum_{j=1}^k p_j x_j}.
\label{eq:convex_comb}
\end{equation}
If $p_0>0$, then the fraction on the right-hand side of \eqref{eq:convex_comb} can be seen as a convex combination of
\[
\left\{-\dfrac{\tilde{p}^\top\kappa\theta}{p_0},\dfrac{\tilde{p}^\top\kappa_1}{p_1},\ldots,\dfrac{\tilde{p}^\top\kappa_{k}}{p_{k}}\right\},
\]
with coefficients $p_0,p_1x_1,\ldots,p_kx_k$. As a consequence, we have in this case
\[
 \displaystyle \sup_{x\in (0,\infty)^d}-\frac{p^\top G_1 H_1(x)}{p^\top H_1(x)}=\max \left\{-\dfrac{\tilde{p}^\top\kappa\theta}{p_0},\dfrac{\tilde{p}^\top\kappa_1}{p_1},\ldots,\dfrac{\tilde{p}^\top\kappa_{k}}{p_{k}}\right\}.
\]
If $p_0=0$, then using the assumption $\kappa\theta\ge 0$ (cfr., Proposition \ref{prop:posSol}) we get
\begin{align*}
\displaystyle \sup_{x\in (0,\infty)^d}-\frac{p^\top G_1 H_1(x)}{p^\top H_1(x)}
&=\displaystyle \sup_{x\in (0,\infty)^k}\frac{-\tilde{p}^\top\kappa\theta+\sum_{j=1}^k \tilde{p}^\top\kappa_j x_j}{\sum_{j=1}^k p_j x_j}\\
&= \displaystyle \sup_{x\in (0,\infty)^k}\frac{\sum_{j=1}^k \tilde{p}^\top\kappa_j x_j}{\sum_{j=1}^k p_j x_j}\\
&=\max \left\{\dfrac{\tilde{p}^\top\kappa_1}{p_1},\ldots,\dfrac{\tilde{p}^\top\kappa_{k}}{p_{k}}\right\}.
\end{align*}

\subsection{Proof of Proposition \ref{prop:LJD_eigenvalues}}
Suppose first that $\kappa$ is lower triangular.
In order to get a specific idea what the matrix $G_2$ looks like, we start by fixing a monomial basis for $\mathrm{Pol}_2(\R^d)$ using the graded lexicographic ordering of monomials:
\begin{align}
H_2(x)=(1,x_1,\ldots,x_d,x_1^2,x_1x_2,\ldots, x_1x_d,x_2^2,x_2x_3,\ldots,x_d^2)^\top,\quad x\in \R^d.\label{eq:ordered_monomials}
\end{align}
It follows by inspection of \eqref{eq:proof_pol_jump_diff} and \eqref{eq:appendix_generator_jump} that, thanks to the triangular structure of $\kappa$, the matrix $G_2$ is lower triangular with respect to this basis. Indeed, the first and third term in \eqref{eq:proof_pol_jump_diff} only contribute to the diagonal elements of $G_2$, while the second term contributes to the lower triangular part (including the diagonal). The eigenvalues of $G_2$ are therefore given by its diagonal elements.

Each element in the monomial basis can be expressed as as $f(x)=x_1^{\alpha_1}\cdots x_d^{\alpha_d}$, for some $\alpha\in\N^d$ with $\sum_{i=1}^d \alpha_i\le 2$. In order to find the diagonal elements of $G_2$, we need to find the coefficient of the polynomial $\Gcal f(x)$ associated with the basis element $f(x)$. It follows from \eqref{eq:proof_pol_jump_diff} and \eqref{eq:appendix_generator_jump} that this coefficient is given by
\begin{align*}
&-\sum_{i=1}^d\kappa_{ii}\alpha_i +\frac{1}{2}\sum_{i<j}(\Sigma\Sigma^\top)_{ij}\alpha_i\alpha_j+\sum_{i=1}^d (\Sigma\Sigma^\top)_{ii}\alpha_i(\alpha_i-1)\\
&+\xi\,\int_\Scal \left(\e^{\alpha^\top \log(1+z)}- 1- \alpha^\top z\right)\, F(\dd z).
\end{align*}
The restriction $\sum_{i=1}^d \alpha_i\le 2$ allows to summarize all diagonal elements, and hence the eigenvalues, of $G_2$ as follows
\begin{gather*}
0, -\kappa_{11},\dots, -\kappa_{dd}, \\
-\kappa_{ii}-\kappa_{jj}+(\Sigma\Sigma^\top)_{ij}+\xi\,\int_\Scal z_iz_j \, F(\dd z),\quad 1\le i,j\le d.
\end{gather*}
Note that a change of basis will lead to a different matrix $G_2$, however its eigenvalues are unaffected. The choice of the basis in \eqref{eq:ordered_monomials} is therefore without loss of generality.

If $\kappa$ is upper triangular, we consider a different ordering for the monomial basis:
\begin{align*}
H_2(x)=(1,x_d,\ldots,x_1,x_d^2,x_d x_{d-1},\ldots, x_d x_1,x_{d-1}^2,x_{d-1} x_{d-2},\ldots,x_1^2)^\top,\quad x\in \R^d.
\end{align*}
The result now follows from the same arguments as in the lower triangular case.

%

\subsection{Proof of Proposition \ref{prop:dividend_swap}}
Using the law of iterated expectations and the moment formula \eqref{eq:moment_formula} we get:
\begin{align*}
\E_t[\zeta_{T_2}(C_{T_2}-C_{T_1})]&=
\e^{-\gamma T_2}\left(\e^{\beta T_2}\E_t[q^\top H_1(X_{T_2})  p^\top H_1(X_{T_2})]-\e^{\beta T_1}\E_t[p^\top H_1(X_{T_1})\E_{T_1}[q^\top H_1(X_{T_2})]]\right)\\
&=\e^{-\gamma T_2}\left(\e^{\beta T_2}w_2^\top \e^{G_2(T_2-t)}H_2(X_t)-\e^{\beta T_1}\E_t[p^\top H_1(X_{T_1})q^\top\e^{G_1(T_2-T_1)} H_1(X_{T_1})]\right)\\
&=\e^{-\gamma T_2}\left(\e^{\beta T_2}w_2^\top \e^{G_2(T_2-t)}H_2(X_t)-\e^{\beta T_1}w_1^\top\e^{G_1(T_1-t)} H_2(X_t)\right).
\end{align*}
Note that the vectors $w_1$ and $w_2$ are unique since the basis elements are linearly independent by definition.
Finally, using the bond price formula \eqref{eq:price_bond} we get
\begin{align*}
D_{fwd}(t,T_1,T_2)&=\frac{1}{\zeta_t P(t,T_2)}\E_t[\zeta_{T_2}(C_{T_2}-C_{T_1})]\\
&=\frac{\e^{\beta T_2}w_2^\top \e^{G_2(T_2-t)}H_2(X_t)-\e^{\beta T_1}w_1^\top\e^{G_1(T_1-t)} H_2(X_t)}{q^\top\e^{G_1(T_2-t)}H_1(X_t)}.
\end{align*}
\clearpage


\clearpage


\begin{table}
\center
\begin{tabular}{rccccc}
\toprule
&February & March & March (oos) & April & April (oos)  \\ \midrule
\multicolumn{1}{l}{Dividend futures (ARE in \%)} \\
1y & 0.602 & 1.460 & 1.156 & 1.770 & 0.821 \\ 
2y & 0.982 & 0.743 & 0.949 & 0.941 & 2.344 \\ 
3y & 0.577 & 0.898 & 1.013 & 0.704 & 1.488 \\ 
4y & 0.434 & 0.437 & 0.456 & 0.784 & 0.926 \\ 
5y & 0.549 & 0.466 & 0.434 & 0.506 & 0.343 \\ 
7y & 1.052 & 0.884 & 1.140 & 0.784 & 2.467 \\ 
9y & 0.901 & 0.738 & 0.843 & 1.129 & 3.819 \\ 
[10 pt] \multicolumn{1}{l}{Interest rate swaps (AE in \%)} \\
1y & 0.003 & 0.004 & 0.005 & 0.005 & 0.004 \\ 
2y & 0.021 & 0.032 & 0.011 & 0.037 & 0.011 \\ 
3y & 0.028 & 0.038 & 0.017 & 0.047 & 0.007 \\ 
4y & 0.025 & 0.029 & 0.032 & 0.042 & 0.020 \\ 
5y & 0.021 & 0.025 & 0.047 & 0.026 & 0.040 \\ 
7y & 0.029 & 0.030 & 0.067 & 0.028 & 0.069 \\ 
10y & 0.044 & 0.043 & 0.061 & 0.063 & 0.073 \\ 
[10 pt] \multicolumn{1}{l}{Dividend option (AE in \%)} & 0.407 & 0.871 & 0.912 & 0.365 & 0.531 \\ 
[10 pt] \multicolumn{1}{l}{Swaption (AE in bps)} & 1.063 & 2.092 & 2.331 & 1.142 & 3.516 \\ 
[10 pt] \multicolumn{1}{l}{Stock option (AE in \%)} &  1.868 & 0.932 & 3.482 & 1.089 & 1.129 \\ 
[10 pt] \multicolumn{1}{l}{Index level (ARE in \%)} & 0.038 & 0.028 & 0.065 & 0.023 & 0.059 \\ 
\bottomrule
\end{tabular}
\caption{Averages of the Absolute Error (AE) and Absolute Relative Error (ARE). The out-of-sample (oos) errors are calculated using the parameters calibrated on the month before.}
\label{table:abs_error}
\end{table}

\begin{table}
\center
\begin{tabular}{lrrr}
\toprule
Parameter & February & March & April\\ \midrule
$\beta$ & 0.0045 & 0.0043 & 0.0016 \\ 
$\kappa_1^D$ & 0.018 & 0.018 & 0.022 \\ 
$\theta^D$ & 0.0013 & 0.0015 & 0.0015 \\ 
$\kappa^I_0$ & 3.1e-04 & 2.7e-04 & 2.4e-04 \\ 
$\kappa^I_1$  & 0.17 & 0.16 & 0.22 \\ 
$\gamma$ & 0.053 & 0.047 & 0.035 \\ 
$\sigma^D$ & 0.12 &  0.10 & 0.11 \\ 
$\sigma^I$ & 0.34 & 0.32 & 0.45 \\ 
$\rho$ & 0.97 &  0.80 & 0.99 \\ 
\bottomrule
\end{tabular}
\caption{Calibrated model parameters using daily prices from February, March, and April 2015.}
\label{table:params}
\end{table}

\begin{table}
\center
\begin{tabular}{lcccccc}
\toprule
& $N=2$& $N=3$& $N=4$& $N=5$& $N=6$ & MC\\ \midrule
Swaption& 0.02 &0.02 &0.02 &0.02 &0.02 &1.01 \\
Dividend option & 0.03 &0.06 &0.14 &0.41 &1.24 &25.88 \\
Stock option & 0.02 &0.03 &0.06 &0.07 &0.11 &3.49 \\ \bottomrule
\end{tabular}
\caption{Computation times (in seconds) needed to price swaptions, dividend options, and stock options using a) the maximum entropy method matching $N$ moments and b) Monte-Carlo simulation with $10^5$
sample paths and weekly discretization. The swaption has a maturity of 3 months and underlying swap of 10 years, the dividend option has a maturity of 2 years, and the stock option has a maturity of 3 months. All options have ATM strike. All computations are performed on a desktop computer with Intel Xeon 3.50GHz CPU and 16GB of RAM memory.}
\label{tabel:computation_time}
\end{table}
\clearpage

\begin{figure}
\center
\begin{subfigure}[b]{0.45\textwidth}
\includegraphics[width=\textwidth]{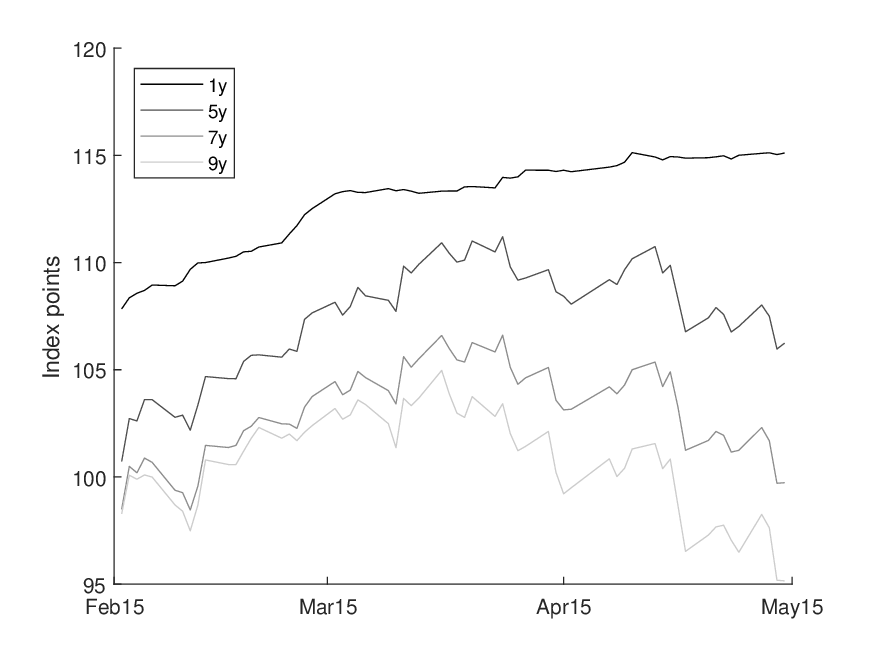}
\caption{Dividend futures}
\label{fig:market_div_fut}
\end{subfigure}
\begin{subfigure}[b]{0.45\textwidth}
\includegraphics[width=\textwidth]{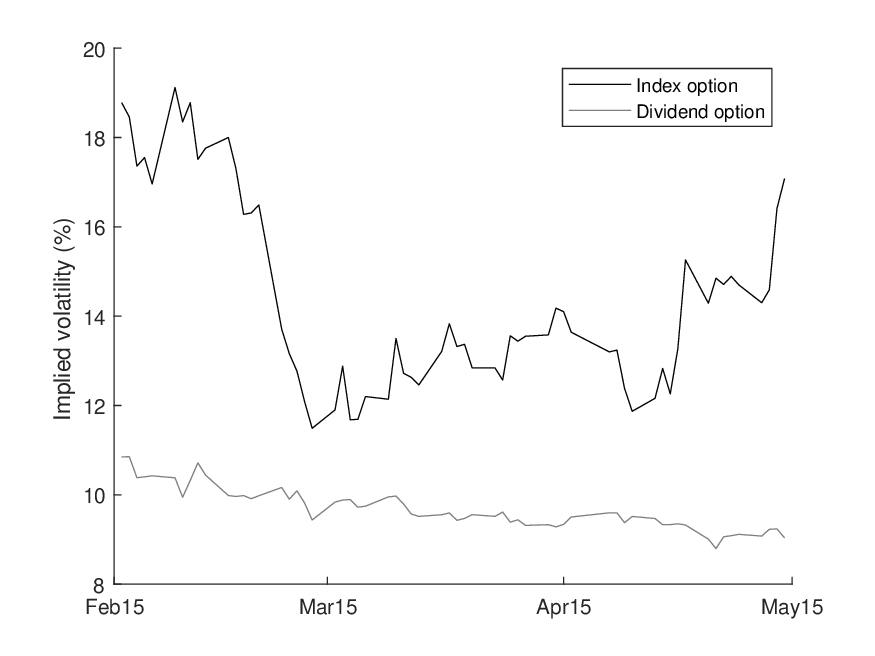}
\caption{Stock option and dividend option}
\label{fig:market_div_eq_opt}
\end{subfigure}
\begin{subfigure}[b]{0.45\textwidth}
\includegraphics[width=\textwidth]{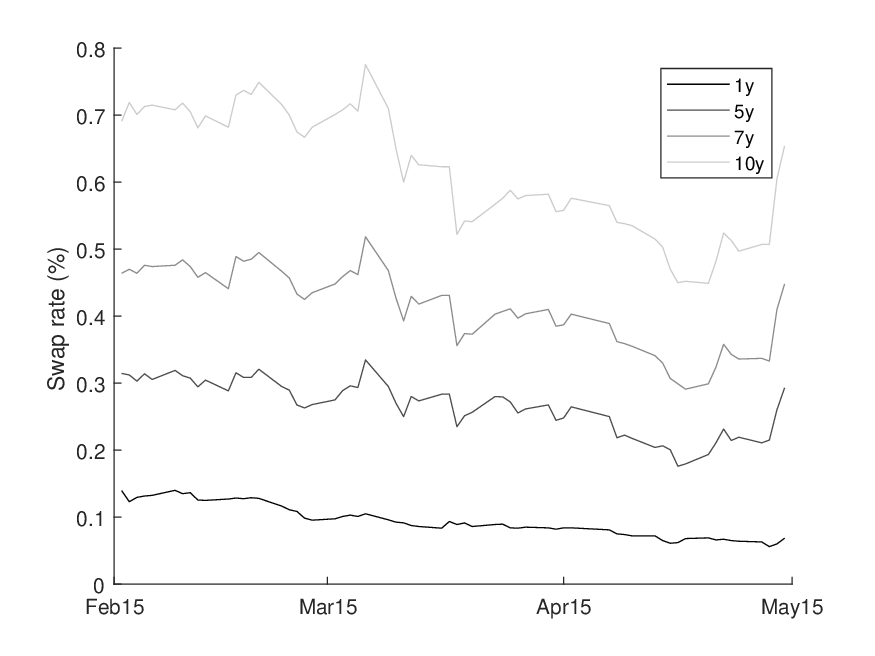}
\caption{Interest rate swaps}
\label{fig:market_IRS}
\end{subfigure}
\begin{subfigure}[b]{0.45\textwidth}
\includegraphics[width=\textwidth]{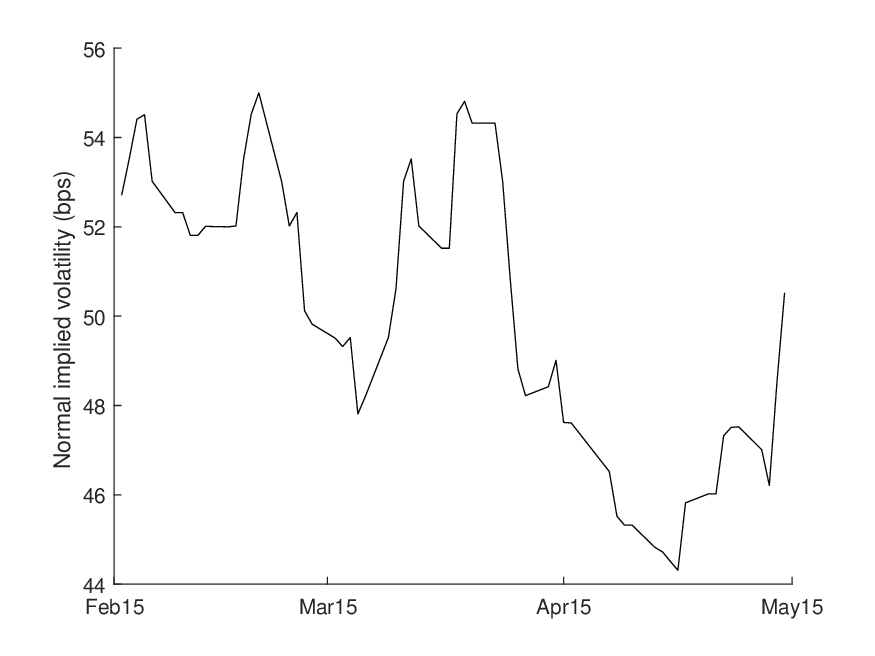}
\caption{Swaption}
\label{fig:market_swpt}
\end{subfigure}
\begin{subfigure}[b]{0.45\textwidth}
\includegraphics[width=\textwidth]{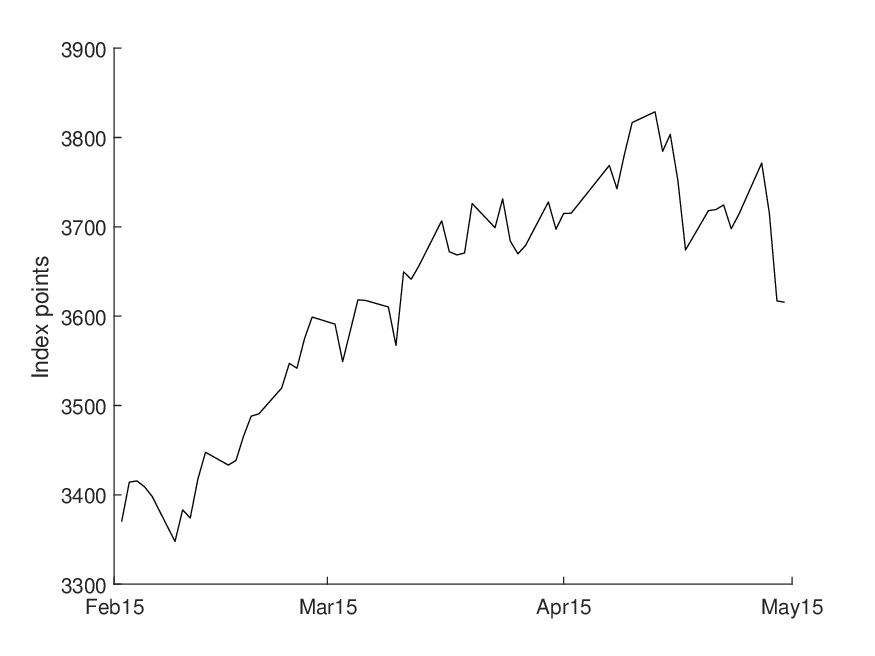}
\caption{Index level}
\label{fig:market_index}
\end{subfigure}
\caption{Data used in the calibration exercise. Dates range from February 2015 until April 2015 at a daily frequency. Figure \ref{fig:market_div_fut} shows the interpolated Euro Stoxx 50 dividend futures prices with a constant time to maturity of 1, 5, 7, and 9 years. The contracts with  time to maturity of 2, 3, and 4 years are not plotted for clarity. Figure \ref{fig:market_IRS} shows the par swap rate of Euribor spot starting swaps with tenors 1, 5, 7, and 10 years. The swap rates with tenors 2, 3, and 4 years are not plotted for clarity. Figure \ref{fig:market_div_eq_opt} shows the Black-Scholes and Black implied volatility, respectively, of ATM Euro Stoxx 50 index and dividend options. The stock option has a time to maturity of 3 months and the dividend option 2 years. Figure \ref{fig:market_swpt} shows the normal implied volatility of swaptions with time to maturity 3 months and the underlying swap has a tenor of 10 years. Figure \ref{fig:market_index} shows the level of the Euro Stoxx 50.}
\label{fig:market}
\end{figure}

\begin{figure}
\center
\begin{subfigure}[b]{0.45\textwidth}
\includegraphics[width=\textwidth]{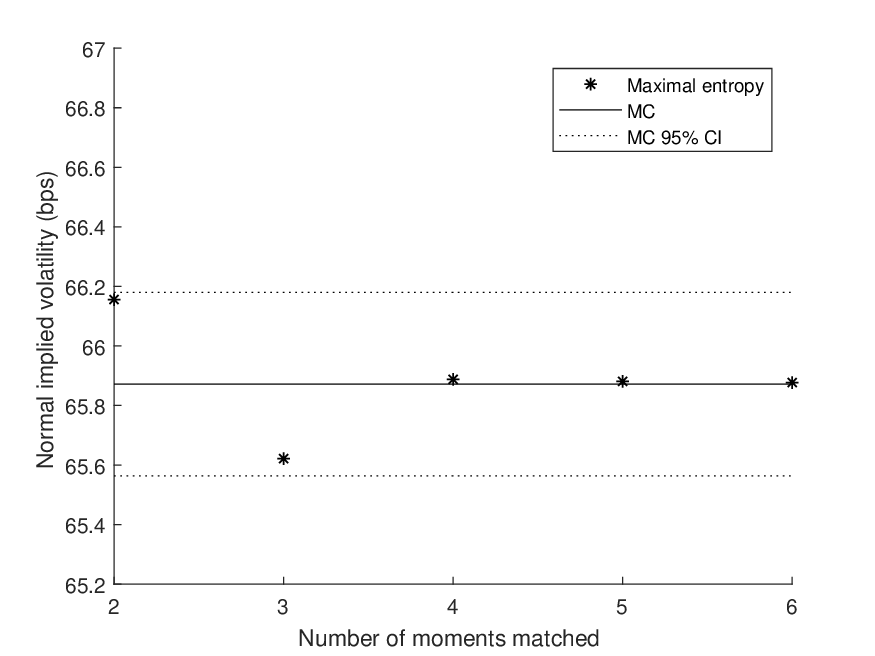}
\caption{Swaption}
\end{subfigure}
\begin{subfigure}[b]{0.45\textwidth}
\includegraphics[width=\textwidth]{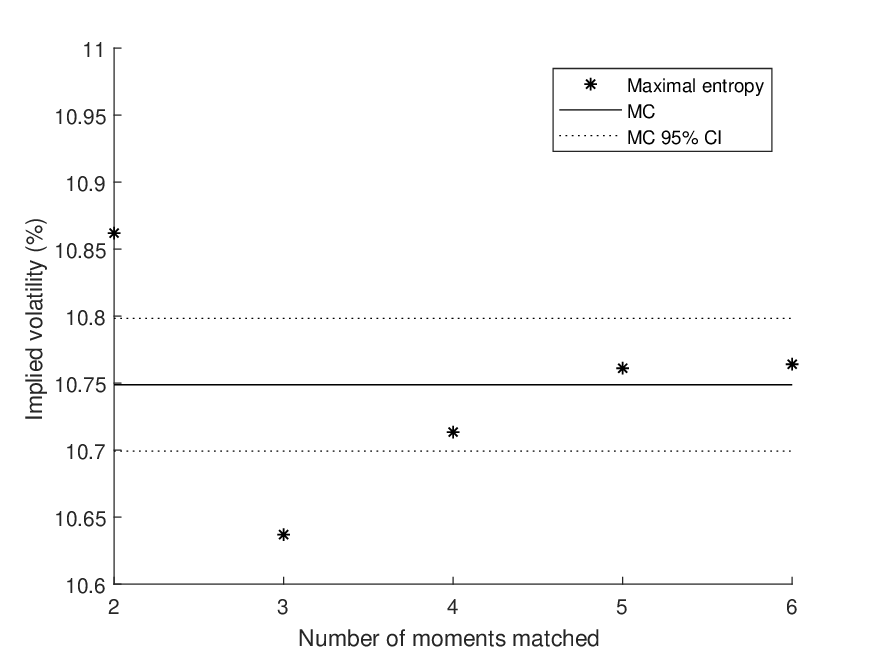}
\caption{Dividend option}
\end{subfigure}
\begin{subfigure}[b]{0.45\textwidth}
\includegraphics[width=\textwidth]{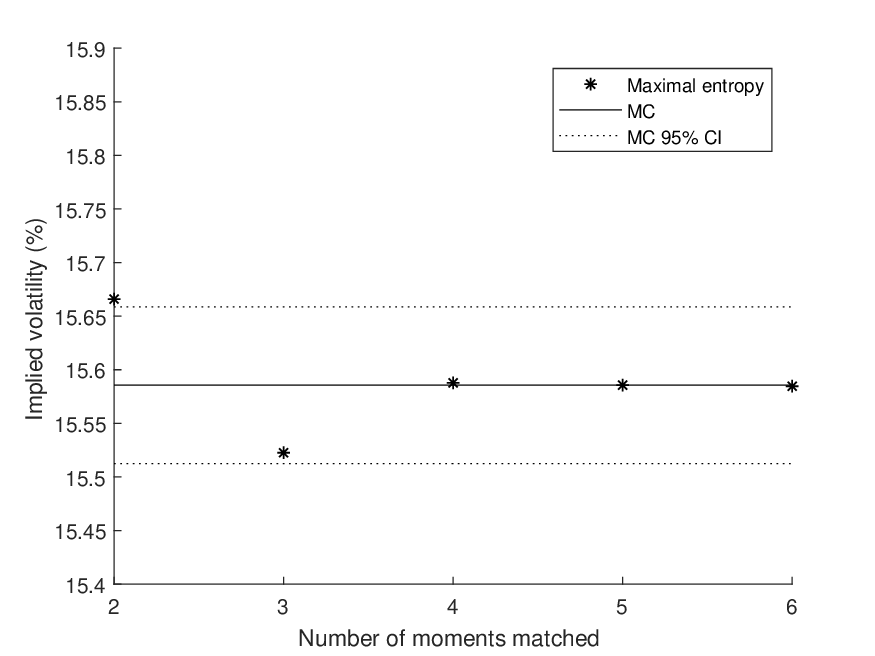}
\caption{Stock option}
\end{subfigure}
\begin{subfigure}[b]{0.45\textwidth}
\includegraphics[width=\textwidth]{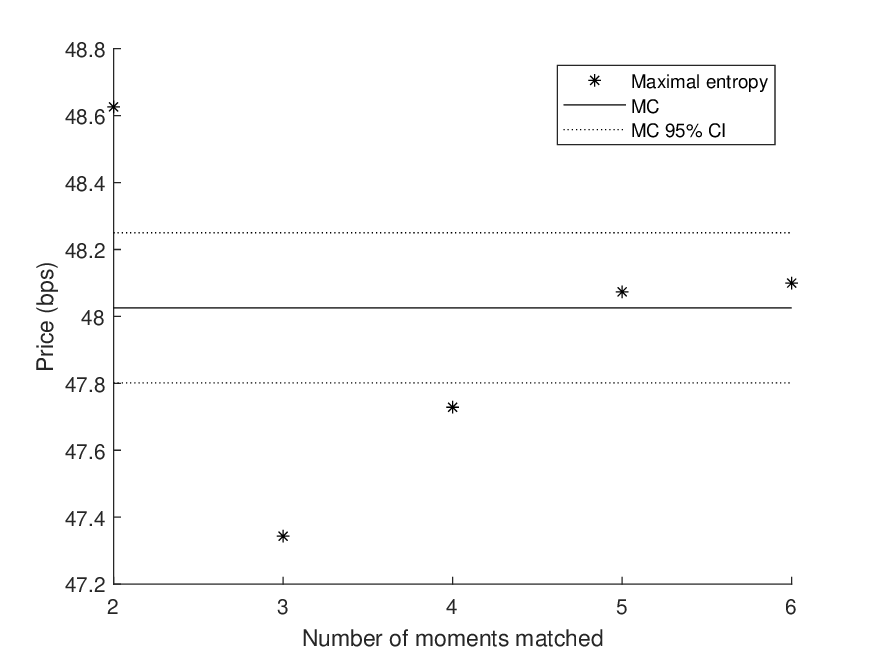}
\caption{Hybrid option}
\end{subfigure}
\caption{Maximum entropy option prices for different number of moments matched. The swaption has maturity 3 months and underlying swap with tenor ten years, the dividend option has maturity 2 years, the stock option has maturity 3 months, and the hybrid option has a single cashflow in 1 year. All options have ATM strike, where we regard the spread $s$ as the strike price for the hybrid option.}
\label{fig:moment_convergence}
\end{figure}

\begin{figure}
\center
\begin{subfigure}[b]{0.45\textwidth}
\includegraphics[width=\textwidth]{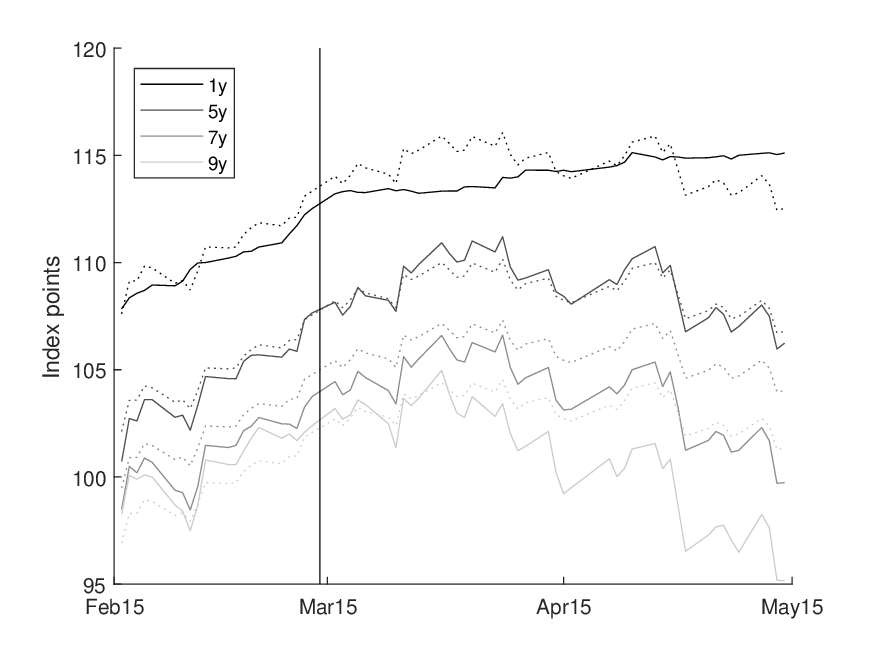}
\caption{Dividend futures}
\label{fig:market_div_fut_fit}
\end{subfigure}
\begin{subfigure}[b]{0.45\textwidth}
\includegraphics[width=\textwidth]{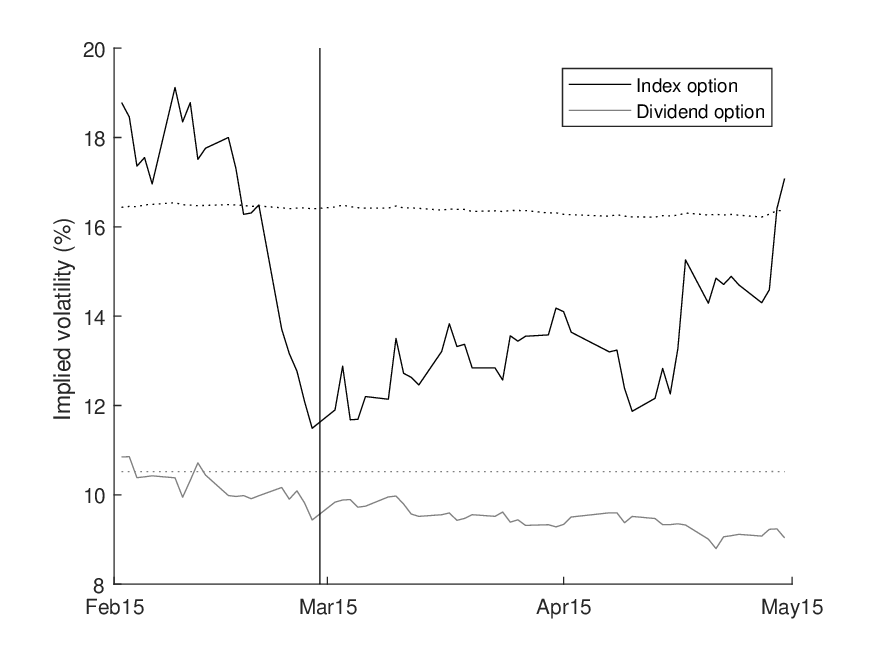}
\caption{Stock option and dividend option}
\label{fig:market_div_eq_opt_fit}
\end{subfigure}
\begin{subfigure}[b]{0.45\textwidth}
\includegraphics[width=\textwidth]{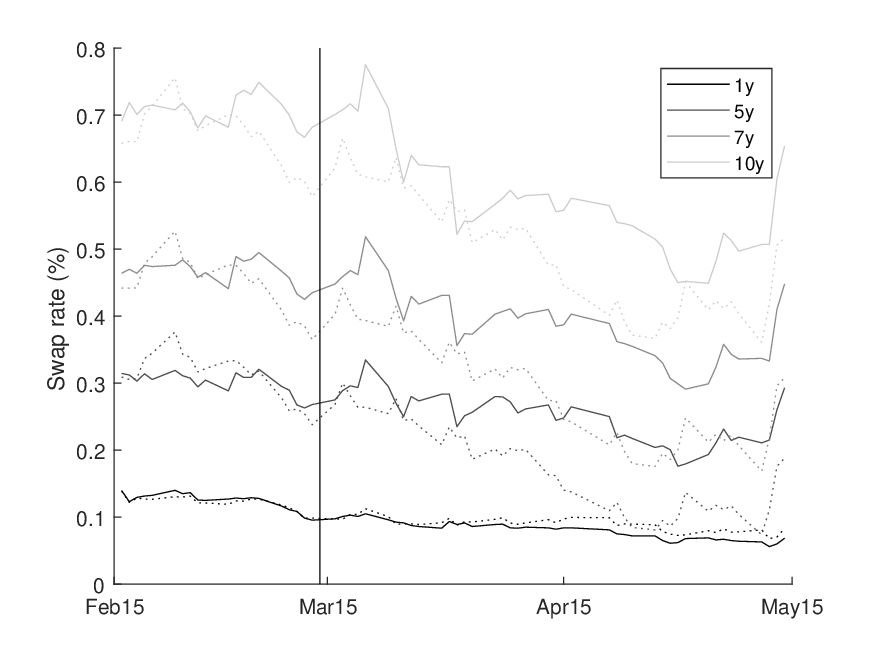}
\caption{Interest rate swaps}
\label{fig:market_IRS_fit}
\end{subfigure}
\begin{subfigure}[b]{0.45\textwidth}
\includegraphics[width=\textwidth]{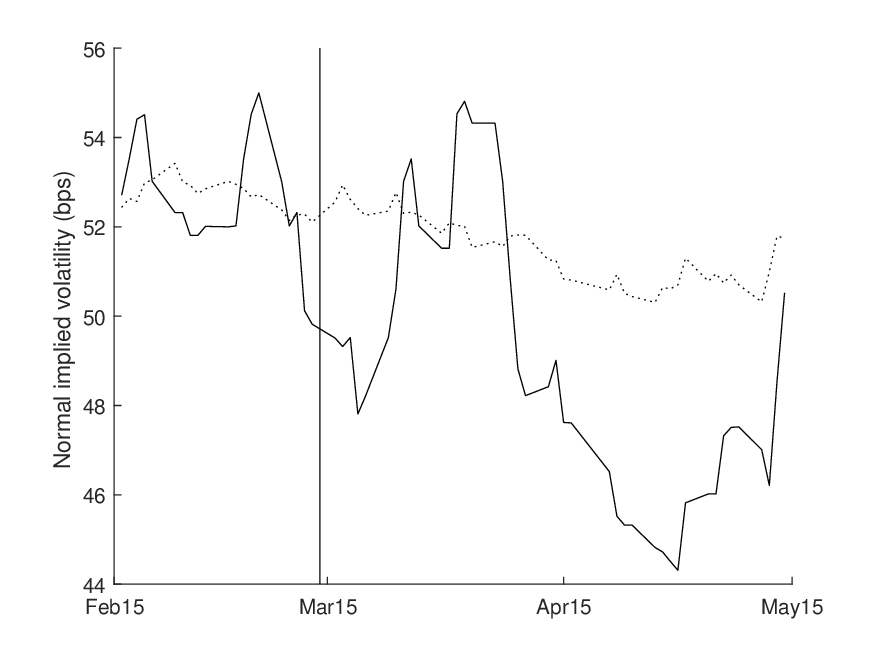}
\caption{Swaption}
\label{fig:market_swpt_fit}
\end{subfigure}
\begin{subfigure}[b]{0.45\textwidth}
\includegraphics[width=\textwidth]{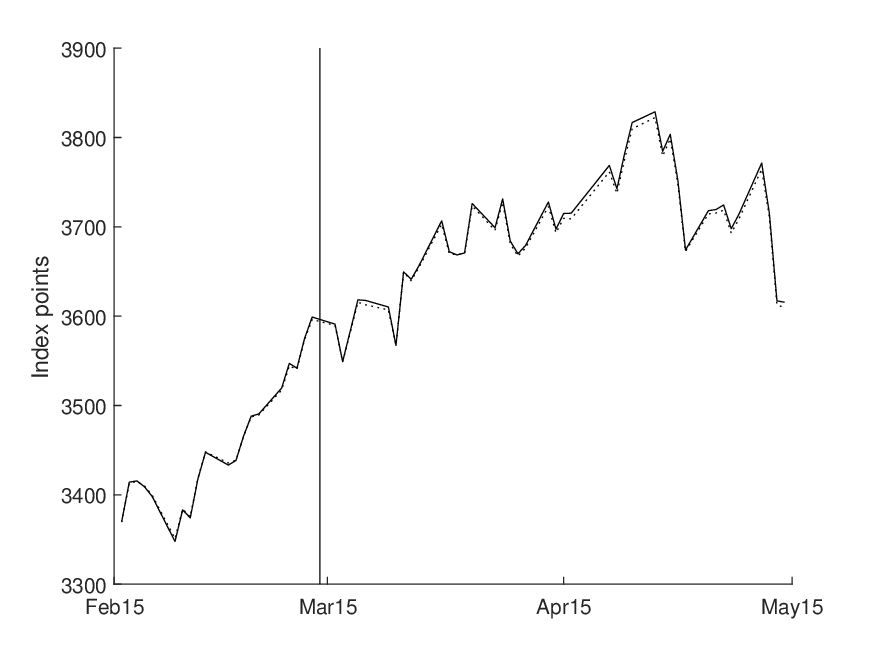}
\caption{Index level}
\label{fig:market_index_fit}
\end{subfigure}
\begin{subfigure}[b]{0.45\textwidth}
\includegraphics[width=\textwidth]{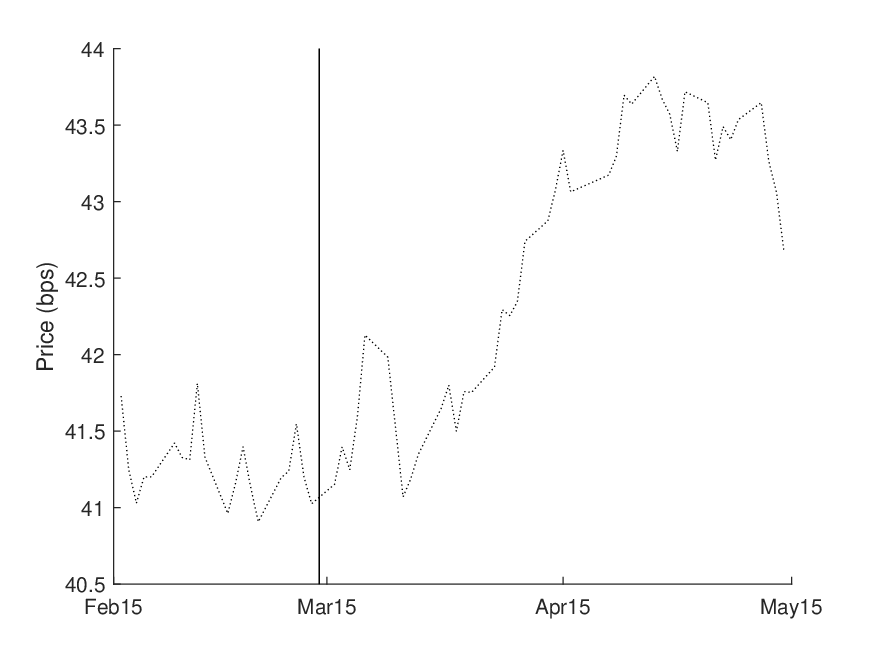}
\caption{Hybrid option price}
\label{fig:hybrid}
\end{subfigure}
\caption{Market prices (solid lines) and model implied prices (dotted lines) using the February parameters. The vertical line indicates the last day of February 2015. The hybrid option price does not have an observable market price.}
\label{fig:market_fit}
\end{figure}

\begin{figure}
\center
\begin{subfigure}[b]{0.45\textwidth}
\includegraphics[width=\textwidth]{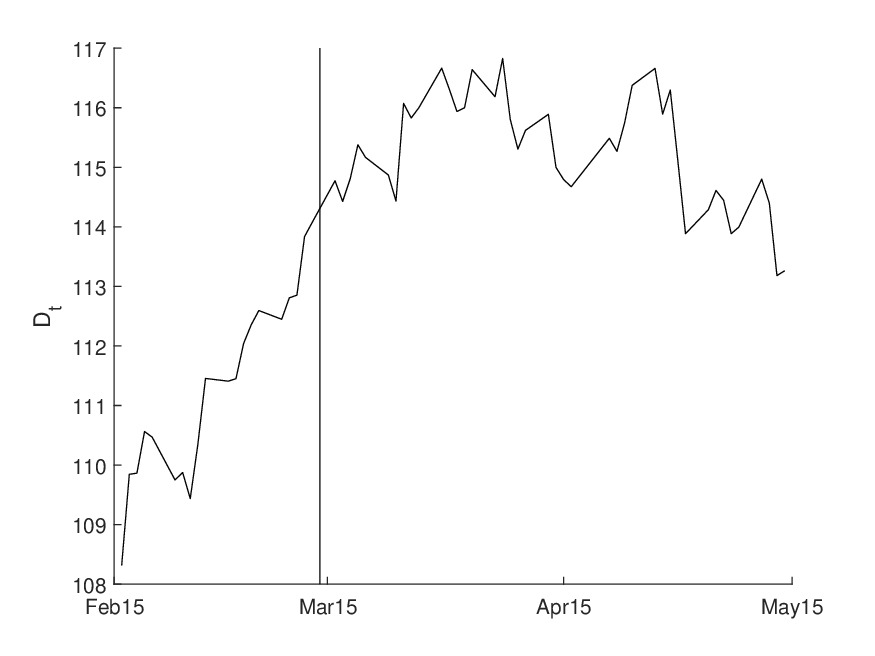}
\caption{Dividend rate $D_t$}
\label{fig:state_div}
\end{subfigure}
\begin{subfigure}[b]{0.45\textwidth}
\includegraphics[width=\textwidth]{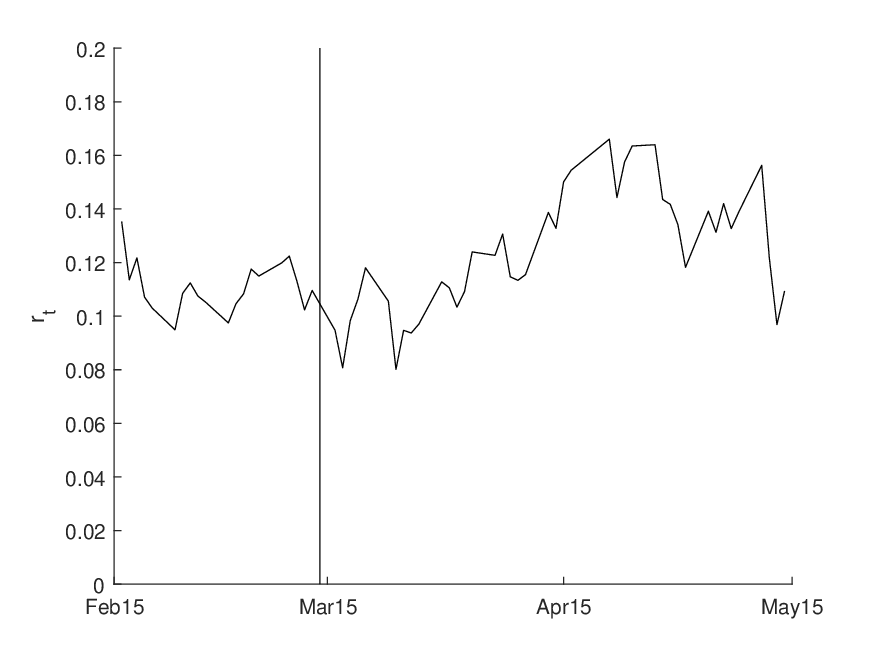}
\caption{Short-rate $r_t$}
\label{fig:short_rate}
\end{subfigure}
\caption{Dividend rate and short-rate using the February parameters. The vertical line indicates the last day of February 2015.}
\label{fig:state}
\end{figure}

\begin{figure}
\center
\includegraphics[width=0.6\textwidth]{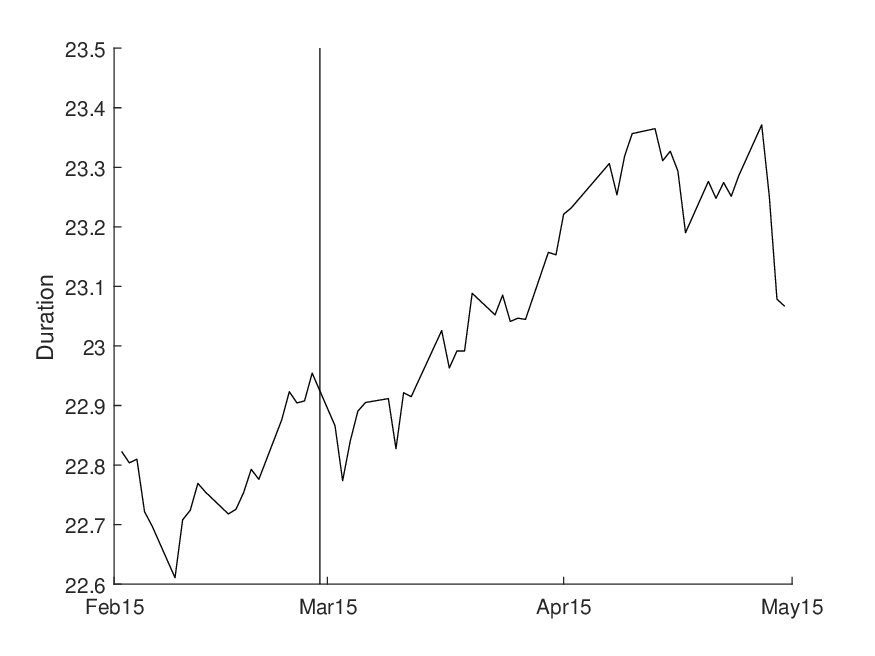}
\caption{Stock duration using the February parameters. The vertical line indicates the last day of February 2015.}
\label{fig:duration}
\end{figure}

\begin{figure}
\center
\includegraphics[scale=0.8]{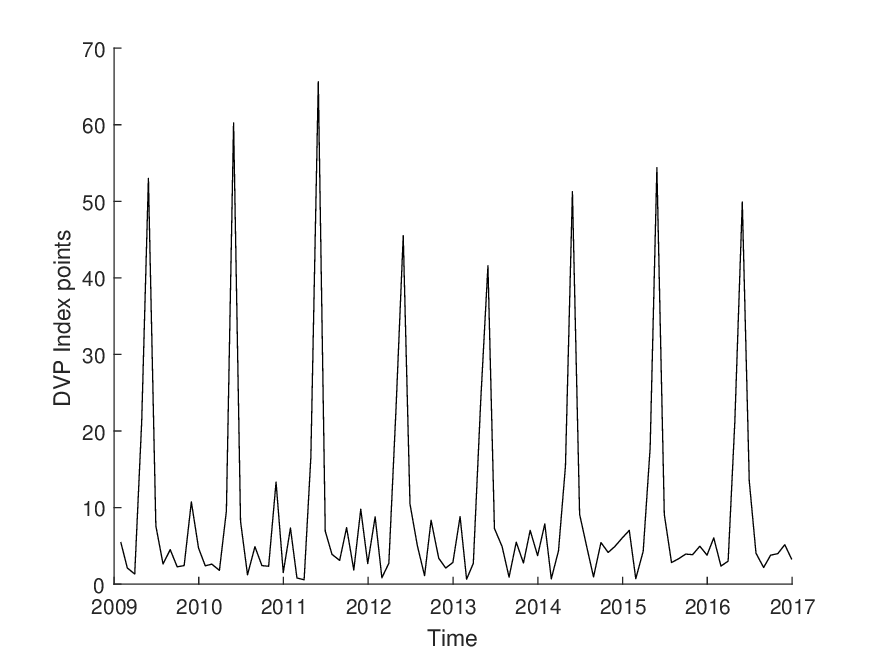}
\caption{Monthly dividend payments by Euro Stoxx 50 constituents (in index points) from January 2009 until December 2016. Source: Euro Stoxx 50 DVP index, Bloomberg.}
\label{fig:seasonality}
\end{figure}

\end{document}